\documentclass[10pt]{article}
\usepackage{amsmath,amsthm,amsfonts,amscd,eucal,latexsym,amssymb,amsbsy,dsfont,mathrsfs,yhmath}
\usepackage{geometry} % see geometry.pdf on how to lay out the page. There's lots.
\newsymbol\bt 1202           %%% \boxtimes
\def\bpi{\pi}

\def\bt{{\bf t}}
\def\vphi{\varphi}

\def\cB{{\ca B}}
\def\cI{{\ca I}}

\def\cH{{\ca H}}
\def\cD{{\ca D}}

\def\cO{{\ca O}}
\def\cS{{\ca S}}

\def\cA{\mathscr{A}}

\def\cW{\mathscr{W}}
\def\mD{\mathscr{D}}
\def\cR{\mathscr{R}}

\def\bC{{\mathbb C}}           %%%  complex numbers and so on

\def\bI{{\mathbb I}}

\def\bN{{\mathbb N}}
\def\bM{{\mathbb M}}
\def\NN{{\mathbb N}}
\def\bR{{\mathbb R}}

\def\bS{{\mathbb S}}

\def\bZ{{\mathbb Z}}
 
\newsymbol \rest 1316         %%% restriction symbol
\def\gA{{\mathfrak A}}       %%% Ghotic
\def\gB{{\mathfrak B}}

\def\gF{{\mathfrak F}}

\def\gH{{\mathfrak H}}

\def\gM{{\mathfrak M}}

\def\gR{{\mathfrak R}}

\def\beq{\begin{eqnarray}}
\def\eeq{\end{eqnarray}}
\newcommand{\ca}[1]{{\cal #1}}         %%  calligraphic

\def\supp{\mbox{supp}}

\newcommand{\Int}{\mathrm{Int}}
\newcommand{\Imm}{\mathrm{Im}}

\newcommand{\ide}{\mathds{1}}

\begin{document}

\title{Topological features of massive bosons on two dimensional Einstein space-time}
\author{{\Large Romeo Brunetti$^{1,2,a}$,  Lorenzo Franceschini$^{1,b}$, Valter Moretti$^{1,2,3,c}$} 
\\\null\\
   \noindent$^1$ Dipartimento di Matematica -- Universit\`a di Trento
\\
 \noindent$^2$ Istituto Nazionale di Fisica Nucleare -- Gruppo Collegato di Trento
\\
\noindent$^3$  Istituto Nazionale di Alta Matematica --  unit\`a locale di Trento\\
 via Sommarive 14  
 I-38050 Povo (TN), Italy. 
 \\
\small  $^a$brunetti@science.unitn.it,
 $^b$lorenzo.franceschini@katamail.com,  $^c$moretti@science.unitn.it}

%%% BEGIN DOCUMENT

%% 
%% 
%% \par 
%% \bigskip 
%% \par 
%% \rm 
%%  
%% %%%%%%%%%%%%%   Title %%%%%%%%%%%%%%%%%%%%%%%%%% 
 
%% \par 
%% \bigskip 
%% \LARGE 
%% \noindent 
%% {\bf Topological features of massive bosons on two dimensional Einstein space-time. I: Spatial approach.} 
%% \bigskip 
%% \par 
%% \rm 
%% \normalsize 
%%  
%%%%%%%%%%%%%%%%%%%%%%%%%%%%%%%%%%%%%%%%%%%%%
%%%%%%%%%%%% Authors %%%%%%%%%%%%%%%%%%%%%%%%
%% 
%% \large
%% \noindent {\bf Romeo Brunetti$^{1,2,a}$},
%% {\bf Lorenzo Franceschini$^{1,b}$}, {\bf Valter Moretti$^{1,2,3,c}$} \\
%% \par
%% \small
%% \noindent $^1$ 
%% Dipartimento di Matematica, Universit\`a di Trento, via Sommarive 14  
%% I-38050 Povo (TN), Italy
%% \smallskip
%Dipartimento di Fisica Nucleare e Teorica, Universit\`a di Pavia and  Istituto Nazionale di Fisica Nucleare  
%Sezione di Pavia, via A.Bassi 6 I-27100 Pavia, Italy.\smallskip
%% 
%% \noindent$^2$ Istituto Nazionale di Fisica Nucleare -- Gruppo Collegato di Trento, via Sommarive 14  
%% I-38050 Povo (TN), Italy. \smallskip
%% 
%% \noindent$^3$  Istituto Nazionale di Alta Matematica ``F.Severi''-- GNFM -- unit\`a locale di Trento,
%%  via Sommarive 14  
%% I-38050 Povo (TN), Italy.\bigskip
%% 
%% \noindent E-mail: $^a$brunetti@science.unitn.it,
%%  $^b$lorenzo.franceschini@katamail.com,  $^c$moretti@science.unitn.it\\ 
%%  \normalsize
%% 
%% 
%% 
%% 

 \theoremstyle{plain}

  \newtheorem{definition}{Definition}[subsection]

  \newtheorem{theorem}[definition]{Theorem}

  \newtheorem{proposition}[definition]{Proposition}

  \newtheorem{corollary}[definition]{Corollary}

  \newtheorem{lemma}[definition]{Lemma}

\theoremstyle{definition}

 \newtheorem{remark}[definition]{Remark}

  \newtheorem{example}[definition]{Example}

\maketitle

\begin{center}
Preprint UTM 725
\end{center}

\begin{abstract}
In this paper we tackle the problem of constructing explicit examples of topological cocycles of Roberts' net cohomology, as defined abstractly by Brunetti and Ruzzi. 
We consider the simple case of massive bosonic quantum field theory on the two dimensional Einstein cylinder. After deriving some crucial results of the algebraic framework of quantization, we address the problem of the construction of the topological cocycles. All constructed  cocycles lead to unitarily equivalent representations of the fundamental group of the circle (seen as a diffeomorphic image of all possible Cauchy surfaces). The construction is carried out using only Cauchy data and related net of local algebras on the circle. 
\end{abstract}

\tableofcontents

\section{Introduction}\label{sec:intro}
The rigorous analysis of features of quantum field theories on curved space-time entered recently in a mature stage. 
 One may consider for instance, the precise description of renormalization for perturbative interacting quantum 
 field theories \cite{BF,HW1,HW2,HW3,HW4,M03,HW4}, especially as the application of the 
 new first principle called \emph{local covariance} \cite{BFV}, the analysis 
 of operator product expansion \cite{HW_OPE,Hol}, the new development of superselection sectors
  \cite{GLRV,RuzziRMP,BR1,BR2}, the studies related to cosmologically important models with the kind of 
  duality (boundary-bulk) effects \cite{DMP,M08,DMP2}, the insights into the energy inequalities \cite{F}, 
  and the analysis of local thermodynamical features \cite{BS,SV}, as the main interesting points 
  of the new results. One expects now a period of expansion towards new developments especially 
  directed towards applications to cosmology of early universe, and to the analysis of specific new features. 

In this paper we discuss one instance of the last mentioned direction. Namely, we wish to consider 
the recent analysis of Brunetti and Ruzzi \cite{BR2} about the encoding of topological features 
of space-time in the superselection structure of quantum field theories, by analyzing the simplest
 possible model, i.e. massive scalar quantum field theory on the two dimensional space-time cylinder. 
 In \cite{BR2}, the authors worked out the general strategy on 4-dimensional space-times, and found a precise 
 description of a new kind of superselection sectors that carry information on the space-time topology. There, the possibility to split  the 1-cocycles in terms of charged and topological parts was crucially employed, and a result on the triviality of the topological part was indeed discovered for the case of Cauchy surfaces with abelian fundamental group. In the case at hand, however, the situation is remarkably different. In fact, even if the fundamental group of the circle (seen as a Cauchy surface) is abelian, nonetheless the topological part is not given by a character of the group. Hence, as it also happens in other situations, the lower dimensional physics seems to be richer than the higher dimensional one. The superselection structure in the traditional situation of Doplicher, Haag and Roberts \cite{DHR} for bosons on two dimensional Minkowski spacetime was determined by M\"uger \cite{MM} for the generic massive case, and recently by Ciolli \cite{FC} in the free massless case.
 
 Our strategy is first to investigate all features of the model that seem to be required for the analysis of superselection sectors of topological nature, as described in \cite{BR2}, then we pass to the construction of 1-cocycles of topological origin. The construction is straightforward, but sometimes besets by technical nuisances, however the direction should be clear enough, and potentially interesting for further research. We mention also a proof of Haag duality along lines different from the traditional approaches, and also in this case potentially fruitful for further generalizations to curved spacetimes.
 
Besides the present section, the next contains essentially notations and some technical aspects, which however are well-known but worth to stress again. The main point we wish to emphasize is that we shall work exclusively in terms of Cauchy data, i.e.~our preferred geometrical arena will be the circle $\bS^1$. In the third section we address ourselves to proving all necessary ingredients, as additivity, duality, split, Borchers' property and several others, that are necessary for the development of the superselection sector strategy, as advertised in \cite{BR2}. The fourth section is the part in which we construct the topological 1-cocycles, and where we show that they lead to non-trivial unitary representations of the fundamental group of the circle. The last two sections form the technical core of the paper, and sometimes, in order to alleviate the reader from the burden of technical details, we shifted the heavier proofs to the appendices.
 
 \section{The algebra and the vacuum}\label{sec:weyl}
In the following $\bM$ is the globally hyperbolic space-time diffeomorphic to $\bS^1
\times \bR$, the Einstein space-time. If $\theta \in [-\pi,\pi]$ (with identified endpoints) is the standard coordinate over
$\bS^1$ and $t\in \bR$, the metric reads
$$g = - dt\otimes dt +  d\theta \otimes d \theta\:.$$
We consider the free quantum field theory on $\bM$ for the real Klein-Gordon field $\vphi$
with mass $m>0$  and equation of motion:
\beq
(-\partial_t^2 + \partial^2_\theta - m^2) \vphi(t,\theta) = 0\:. \label{KG}
\eeq
Let us fix from now on $\Sigma \equiv \bS^1 \equiv [-\pi,\pi]$ (with identified endpoints), a
space-like
smooth Cauchy surface of $\bM$, normal to $\partial_t$.  Notice that with our convention
the length of $\Sigma\equiv \bS^1$ is $2\pi$.
For future convenience we also fix a {\em positive rotation} as the counterclockwise orientation for $\bS^1$.\\
We remind the reader that a {\em proper interval} of $\bS^1$ is a connected subset $I\subset \bS^1$
 such that both internal parts $\Int(I)$ and $\Int(\bS^1\setminus I)$ are nonempty.
The class of open proper intervals of
 $\bS^1$ will be denoted by  $\cR$. Causality will be understood as disjointness of intervals since we define the causal complement of $I\in\cR$ as $I'\doteq\Int(\bS^1\setminus I)$. Notice that $\cR$ is left invariant by the operation of causal complementation. Later on we shall also employ the notation $I\pm (-\epsilon,\epsilon)\doteq (\inf I\mp\epsilon, \sup I\pm \epsilon)$. \\

% \remark \\
% {\bf (1)}  The support property of Cauchy data of solutions of Klein-Gordon equation $(\Phi,\Pi) \in \cS_{I}$
% for some $I \in \cR$ depends on the Cauchy surface: If
% the Cauchy data of $\vphi$ satisfy the requirement about the existence of $I$ with respect to $\Sigma$,
% Cauchy data
% of the same solution $\vphi$ generally fail to fulfill the requirement
% when referring to another Cauchy surface sufficiently far in time from the former. \\

 \subsection{The Weyl algebra} 
 \noindent As known, the Cauchy problem for the normal hyperbolic partial differential equation \eqref{KG} 
in a globally hyperbolic space-time is well-posed, and we indicate by $\cS$  the real vector space of
 pairs of  smooth functions  $\Phi, \Pi : \Sigma \to \bR$, viewed as
 Cauchy data $\Phi = \vphi\rest_\Sigma$,
 $\Pi = \partial_t \vphi\rest_\Sigma$
 for smooth solutions $\vphi$ of \eqref{KG}. More details on that will be presented in the next section.
If $I \in \cR$, $\cS_{I}$ denotes the subspace of
$\cS$ of the Cauchy data with supports in $I$. It is clear that the space generated by all $\cS_{I}$ is $\cS$ itself,
 since every element of $(\Phi, \Pi) \in \cS$
 can be re-written
as $(\Phi_1,\Pi_1) + (\Phi_2,\Pi_2)$ where $(\Phi_i,\Pi_i) \in \cS_{I_i}$ for $I_i \in \cR$ with
$I_1 \cup I_2 = \bS^1$. This is obtained by
using two functions $\chi_i \in C^\infty(\bS,\bR)$ with $\chi_1 + \chi_2 = 1$ and $\supp \chi_i \subset I_i$ and defining:
$\Phi_i = \Phi \cdot \chi_i$, $\Pi_i = \Pi \cdot \chi_i$.

\noindent  $\cS$ becomes a real symplectic space when equipped with the
  symplectic form  $\sigma : \cS \times \cS \to \bR$  defined by:
 \beq
 \sigma\left((\Phi,\Pi), (\Phi', \Pi')\right) \doteq \int_\Sigma (\Phi' \Pi - \Phi \Pi') d\theta\:.\label{sigma}
 \eeq
 
Referring to the given definitions of $\cS, \cS_{I}$,
  the symplectic form $\sigma$
 and all the symplectic forms obtained by  restricting $\sigma$  to the relevant spaces, are non-degenerate, the
 proof being trivial. As a consequence  there is a unique (see \cite{BRII}), up to $*$-isomorphism, 
 unital $C^*$ algebra $\cW$
associated with $(\cS,\sigma)$  generated by (non-vanishing) Weyl generators $W(\Phi,\Pi)$ satisfying the
standard {\em Weyl relations}: $\forall (\Phi, \Pi), (\Phi_i,\Pi_i) \in \cS$,
\begin{align*}
W(-\Phi,-\Pi)&= W(\Phi,\Pi)^*\, ,\\
 W(\Phi_1,\Pi_1)W(\Phi_2,\Pi_2) &=
 W(\Phi_1+\Phi_2,\Pi_1+\Pi_2)\:\exp\left\{i\sigma((\Phi_1,\Pi_1),(\Phi_2,\Pi_2))/2\right\}\ .
 \end{align*}
  $\cW$ is called the {\em Weyl algebra} associated with $(\cS,\sigma)$. \\
 
 \begin{remark}\label{remark2}
 {\bf (1)} Consider the class $\{\cW(I)\}_{I\in \cR}$
  where $\cW(I)$  is the
 Weyl algebra   generated by the $W(\Phi,\Pi)$ with $(\Phi,\Pi)$ supported in $I$.
 Each $\cW(I)$  is in fact a
 sub  $C^*$-algebra of $\cW$.
 $\{\cW(I)\}_{I\in \cR}$  is by no means a {\em net} of $C^*$ algebras (it will be termed a precosheaf, later on) because the class $\cR$ is not directed with respect the partial ordering relation given by the inclusion
 and thus it is not possible to take the (strict) inductive limit defining the overall quasi local
 ($C^*$-) algebra containing every $\cW(I)$. Notice that, however,  all the sub-algebras
 share the same unit element and the following two properties are valid:
 \begin{description}
 \item[]{\hskip1truecm\em Isotony}:\  $\quad\cW(I) \subset \cW(J)\ ,$ $\qquad I \subset J$\ .
 \item[]{\hskip1truecm\em Locality}:\ $\quad\left[ \cW(I), \cW(J) \right] =0\ ,$  $\qquad I\cap J = \emptyset$\ .
 \end{description}
 {\bf (2)} Following Fredenhagen \cite{Fredenhagen}, one can use another construction to define an algebra 
 that replaces the quasi local one in the inductive limit case. It is termed the {\em universal algebra} $\cA$. One may wonder what is the relation with the global Weyl algebra $\cW$. It is shown in Appendix \ref{AppB} that indeed $\cA \equiv \cW$.\\
{\bf (3)} $\Sigma$ is (metrically) invariant under the action of $\bR$ viewed as $\Sigma$-isometry group: $r\in \bR$
induces the isometry
$\beta_r :\theta \mapsto \theta + r$.
If the pull-back $\beta_r^*$ is defined as $(\beta^*_rf)(\theta) \doteq f(\theta-r)$ for all $f\in C^\infty(\bS^1,\bR)$,
the $\Sigma$-isometry group $\bR$  can be represented in terms of a (strongly continuous) one-parameter group
of $*$-automorphisms of $\cW$,
 $\{\alpha_r\}_{r\in\bR}$,
which is uniquely induced
by \beq \alpha_r \left( W(\Phi,\Pi)\right) \doteq 
W\left((\beta^*_r \Phi, \beta^*_r \Pi)\right) \:, \quad \mbox{for all $r\in \bR$\ ,
$(\Phi,\Pi) \in \cS$\ ,}\label{isometries}\eeq
 The existence of such $\{\alpha_r\}_{r\in\bR}$ follows immediately (see Prop.~5.2.8 in \cite{BRII})
from the fact that $\sigma$ is invariant under every $\beta_r^*$.\\
{\bf (4)} Let $\vphi$ be a real smooth solution  of Klein-Gordon equation (\ref{KG})
and take $s\in \bR$. $\vphi_s$ denotes the other
smooth solution ``translated into the future'' by an interval of time $s$,
in the sense that $\vphi_s(t,\theta) = \vphi(t-s,\theta)$ for all
$t\in \bR$ and $\theta\in \bS^1$. Notice that $\vphi_s$ is another solution of Klein-Gordon
equation because the space-time is static.
Passing to the Cauchy data (on the {\em same} Cauchy surface at $t=0$),
 this procedure induces a one-parameter
group of transformations from $\mu_s: \cS\to \cS$ such that $\mu(\Phi,\Pi)$
 are Cauchy data of $\vphi_s$ when $(\Phi,\Pi)$
are those of $\vphi$. $\{\mu_s\}_{s\in \bR}$ preserves the symplectic form due
 to the invariance of the metric under
time displacements. As a consequence we have a (strongly continuous) one-parameter group of $*$-isomorphisms, $\{\tau_s\}_{s\in \bR}$
acting on $\cW$  and 
uniquely defined by the requirement \beq \tau_s \left( W(\Phi,\Pi)\right) \doteq
 W\left(\mu_s (\Phi,\Pi)\right) \:, \quad \mbox{for all $s\in \bR$\ ,
$(\Phi,\Pi) \in \cS$\ .}\label{isometries2}\eeq
{\bf (5)} Solutions of KG equation  with Cauchy data in $\Sigma$ supported in $I\in \cR$ propagate in $\bM$ inside
the subset of $J^+(I) \cap J^-(I)$ as is well known. Therefore one concludes that,
if $(\Phi,\Pi)$ is supported in $I \in \cR$,
$\mu_s (\Phi,\Pi)$ is supported in the interval $I_s \subset \bS^1$ constructed as follows.
Passing to the new variable $\theta'\doteq \theta +c$ for some  suitable constant $c\in \bR$,
one can always represent $I$ as $(-a,a)$ with $0<a<\pi$. In this representation  $I_s \doteq (-a-|s|, a+|s|)$
taking the identification
$-\pi \equiv \pi$ into account. Notice in particular that, for $I\in \cR$, one has  $I_s \in \cR$ if and only if  $|s| <\pi - \ell(I)/2$
(where $\ell(I)$ is the length of $I \in \cR$ when $\ell(\bS^1) = 2\pi$). Whereas it turns out $I_s = \bS^1$
whenever $s> \pi - \ell(I)/2$. \\
{\bf (6)} The groups $\{\alpha_r\}_{r\in\bR}$ and $\{\tau_s\}_{s\in \bR}$ can be combined
into an Abelian group
of $*$-automorphisms $\{\gamma_{(r,s)}\}_{(r,s)\in \bR^2}$ of $\cW$ with
$\gamma_{(r,s)} \doteq \alpha_r \circ \tau_s$.
This group represents the
action of the unit connected-component of Lie group of  isometries of the space-time on the Weyl algebra
associated with the quantum field.
\end{remark}

\subsection{Vacuum representation}\label{sec:vacuum}
  In the complex Hilbert space $L^2(\bS^1, d\theta)$
 define the positive symmetric operator: $$-\frac{d^2}{d\theta^2} + m^2\ide : \: C^\infty(\bS^1, \bC) \to L^2(\bS^1, d\theta)\:.$$
 It is essentially self-adjoint since $C^\infty(\bS^1,\bC)$ contains a dense set of analytic vectors
 made of exponentials $\theta \mapsto e^{in\theta}$, $n\in \bZ$, which are the eigenvectors of the operators.
 The unique self-adjoint extension of this operator, i.e. its closure, will be denoted by
$A : \cD(A) \to L^2(\bS^1, d\theta)$. 
 Notice that $A$ is strictly positive (being $m>0$) and  thus its real powers $A^{\alpha}$,
   $\alpha \in \bR$, are well-defined. The following properties are easily proved.

 \begin{proposition} \label{P1} The  operators $A^\alpha : \cD(A^\alpha) \to L^2(\bS^1,d\theta)$ for $\alpha\in \bR$
 satisfy the following:
 \begin{itemize}
\item[$(a)$] $\sigma(A^\alpha) = \{(n^2+ m^2)^\alpha\ | \ n\in\bZ\}$.
\item[$(b)$] $A^\alpha$ commutes with the standard conjugation $C: L^2(\bS^1, d\theta) \to L^2(\bS^1, d\theta)$
with $(Cf)(\theta) \doteq \overline{f(\theta)}$ furthermore
$A^\alpha (C^\infty(\bS^1, \bR)) = C^\infty(\bS^1,\bR)$ and $\overline{Ran(A^\alpha)} = L^2(\bS^1, d\theta)$.
\item[$(c)$]  If $\alpha \leq 0$, $A^{\alpha} : L^2(\bS^1, d\theta) \to \cD(A^\alpha)$
 are bounded with $||A^{-\alpha}|| = m^{2\alpha}$.
 \end{itemize}
 \end{proposition}

 The $\bR$-linear map $K : \cS \to L^2(\bS^1,d\theta)$ we shall introduce, turns out to be useful
 to determine a preferred unitary irreducible (Fock) representation of Weyl algebra
 called the {\em vacuum representation}.
 We define, for every $(\Phi,\Pi) \in \cS$:
 \beq
 K(\Phi,\Pi) \doteq \frac{1}{\sqrt{2}} \left(A^{1/4} \Phi + i A^{-1/4} \Pi \right)\:. \label{cfourier}
 \eeq
  A natural physical way to understand $K$ is noticing that the solution of (\ref{KG}) with Cauchy data $(\Phi,\Pi) \in \cS$,
 interpreting the derivative w.r.t. time in the sense of $L^2(\bS^1,d\theta)$ topology can be written as
 \beq 
 \phi(t,\cdot) = \frac{1}{\sqrt{2}}e^{-i tA^{1/2}} A^{-1/4} K(\Phi,\Pi) + C 
 \frac{1}{\sqrt{2}}e^{-i tA^{1/2}} A^{-1/4} K(\Phi,\Pi)\ , \label{solform}
 \eeq
 $C : L^2(\bS^1,d\theta) \to L^2(\bS^1,d\theta)$ being the standard complex conjugation.
  The proof is a trivial consequence of Stone theorem and (c)
 of Proposition (\ref{P1}).
 The right-hand side of  (\ref{solform}) turns out to be  $(t,p)$-jointly  smooth and the derivative w.r.t. time coincides with
 that in the $L^2$ sense  \cite{Wald}. Thus, by the uniqueness
 theorem  for solution of Klein-Gordon equation with compactly supported data in globally hyperbolic space-times,
 varying $t\in \bR$ the right-hand side
 of (\ref{solform}) defines the proper solution
individuated by smooth compactly supported Cauchy data $(\Phi,\Pi)$. From (\ref{solform}), interchanging $A^{\pm 1/4}$ with $e^{-i tA^{1/2}}$ it arises that $A^{1/2}$can be seen as the Hamiltonian generator of Killing time displacements, acting on the Hilbert space  of the wave functions $K(\Phi,\Pi)$ associated
with the classical solutions with Cauchy data $(\Phi,\Pi)$.
That Hilbert space is the so called {\em one-particle space}.
This is the central point of view necessary to understand the
construction presented in Theorem \ref{thoremvacuum1} from a physical point of view.
The following fundamental statement about the range of $K$ holds.

 \begin{proposition}\label{H1} With the given definition for $\cS$, and $K$ the following facts are valid.
\begin{itemize}
\item[$(a)$] The range of $K$  is dense in $L^2(\bS^1, d\theta)$.
\item[$(b)$] For every pair $(\Phi, \Pi), (\Phi',\Pi') \in \cS$ it holds
\beq
-\frac{1}{2}\sigma\left((\Phi, \Pi), (\Phi',\Pi') \right) =
\Imm \left\langle K(\Phi, \Pi) , K(\Phi', \Pi')\right\rangle\:, \label{Ksigma}
 \eeq
as a consequence $K$ is injective.
\end{itemize}
\end{proposition}
 
%  \noindent {\em Proof.}
%  (a) is a straightforward consequence of (c) in Proposition \ref{P1}.
%  The proof of (b) is obtained by directed inspection. Injectivity of $K$ is immediate consequences of
%  (\ref{Ksigma})  and
%  non-degenerateness of $\sigma$.
%$\Box$\\

 \noindent Let us construct the {\em vacuum GNS representation} using Proposition \ref{H1}. Let us remind some terminology. In the following, if $\{\alpha_g\}_{g\in G}$ is a representation
of a group $G$ in terms
of $*$-automorphisms of a unital  $*$-algebra $\gA$, a state $\lambda: \gA \to \bC$ will be said to be {\em invariant} under
$\{\alpha_g\}_{g\in G}$
if one has $\lambda\left(\alpha_g(a)\right) = \lambda(a)$ for all $g\in G$ and $a\in \gA$.
 Moreover a representation
$\{U_g\}_{g\in G}$ where every $U_g$ is a unitary operator defined over the GNS Hilbert $\gH_\lambda$ space of $\lambda$,
is said to {\em implement}
$\{\alpha_g\}_{g\in G}$ if
$\bpi_\lambda\left(\alpha_g(a)\right) = U_g \bpi_\lambda(a)  U^*_g$ for all $g\in G$ and $a\in \gA$,
where $\bpi_g$ is the GNS representation of $\gA$.

 \begin{theorem}\label{thoremvacuum1} With the given definition for $\cS$ and $K$, the following facts are valid.\\
\begin{itemize}
\item[$(a)$] There is a pure state $\omega_0 : \cW \to \bC$
 uniquely induced by linearity and continuity from
 \beq
 \omega_0(W(\Phi,\Pi)) = e^{-\frac{1}{2}\langle K(\Phi,\Pi),K(\Phi,\Pi)\rangle} \: \quad \mbox{for all $(\Phi,\Pi)\in \cS$}
 \label{omega}\:.
 \eeq
\item[$(b)$] The GNS representation of $\omega_0$, $(\gH_0, \bpi_0, \Psi_0)$
  is constructed as follows (up to unitarities):
 \begin{itemize}
\item[$(i)$] $\gH_0$ is the symmetrized Fock space with one-particle space $\cH \doteq L^2(\bS^1,d\theta)$;
 \item[$(ii)$] the  representation $\bpi_0$ is isometric and is induced, by linearity and continuity by:
 \beq
 \bpi_0(W(\Phi,\Pi)) = e^{\overline{a(K(\Phi,\Pi)) -  a^*(K(\Phi,\Pi)) }}\:, \label{expaa}
% e^{-i\overline{\sigma((\hat{\Phi},\hat{\Pi}),(\Phi,\Pi))}}
 \eeq
 where $a(K(\Phi,\Pi))$, $a^*(K(\Phi,\Pi))$ are  standard creation and annihilation operators
 (the latter anti-linear in its argument) defined in the dense subspace
 spanned by  vectors with finite number of particles.
 \item[$(iii)$] the cyclic vector $\Psi_0$ is the vacuum vector of $\gH$. 
 \end{itemize}
 \item[$(c)$] $\omega_0$ is invariant under $\{\gamma_{(r,s)}\}_{(r,s)\in \bR^2}$ where
 $\{\gamma_{(r,s)}\}_{(r,s)\in \bR^2}$ is the Abelian group of $*$-auto\-mor\-phisms
 representing the natural
 action of the unit connected-component Lie group of isometries of $\bM$.
 \item[$(d)$] The unique
 unitary representation $\{U_{(r,s)}\}{(r,s)\in \bR^2}$ on $\gH$
  leaving $\Psi$ invariant  and
implementing $\{\gamma_{(r,s)}\}_{(r,s)\in \bR^2}$ fulfills,
 for all $(r,s)\in \bR^2$: $$U_{(r,s)} = e^{-irP^\otimes} e^{isH^\otimes} = e^{-i(rP^\otimes -s H^\otimes)}\:,$$
where the generators $P^\otimes$, $H^\otimes$
 are respectively given by  the tensorialization of the  operators  $P$, $A^{1/2}$ on $\cH$  with
 $P$ given by  the unique self-adjoint extension of $-i\frac{d}{d\theta} : C^\infty(\bS^1,\bC) \to L^2(\bS^1, d\theta)$.
 \end{itemize}
 \end{theorem} 
% \noindent{\em Sketch of Proof.}  (a) and (b) are immediate consequences of Proposition 3.1 and Lemma A.2 in \cite{KayWald}.
% The fact that $\omega$ is pure follows form the cited propositions by the fact that
% $\overline{K(\cS)} = L^2(\bS^1,d\theta) =  \cH$ as a consequence of Proposition \ref{H1}.
% Invariance of $\omega$ under $\gamma$ arises by direct inspection. A known theorem (see \cite{Araki})
% establishes, as a consequence, the existence of a unique unitary representation of $\gamma$ which implements $\gamma$ leaving
% $\Psi$ invariant. 
% The remaining statements are based on standard arguments valid for QFT in static space-times and trivial
% generalizations of the analogs in Minkowski space-time.
%  $\Box$\\
 
 The proof is based on standard arguments that the reader can find in
 the literature, e.g. \cite{Araki,KayWald}, and is therefore omitted.

\section{Araki-Haag-Kastler pre-cosheaves}\label{sec:HK}
In this section we focus on some remarkable properties of the vacuum GNS representation $(\gH_0,\bpi_0, \Psi_0)$
  and for the associated  class $\gR$ of von Neumann algebras
$\{\gR(I)\}_{I\in \cR}$,
where $\gR(I)\doteq \pi_0(\cW(I))''$. In the following $\gR(\bS^1) \doteq \pi_0(\cW)''$
 and $\cB(\gH_0)$ will denote the algebra of all bounded operators on the  Hilbert spaces $\gH_0$.
The class $\gR$  is called a {\em pre-cosheaf} (see for instance \cite{GL92}) since it inherits from $\cW$ the {\em isotonic} property, i.e.,
\begin{equation*}
\gR(I) \subset \gR(J)\quad \text{if}\ I \subset J\ .
\end{equation*}
As a matter of fact, the following properties also are inherited from those of $\cW$,
\begin{alignat*}{4}
%&\text{{\bf Isotony}}\qquad  &\gR(I) &\subset \gR(J)\quad \text{if}\ I \subset J\ ,\\
&\text{{\em Locality}}\qquad  &\left[ \gR(I), \gR(J) \right] &=0 \ ,
 \qquad \text{if}\ I\cap J = \emptyset\ ,\\
&\text{{\em Covariance}}\qquad & \beta_{r}\gR(I) & = \gR(I_{r})\ ,
\end{alignat*}
where $\beta$ is the lift to $\gR$ of the group of translations on $\bS^1$ acting on $\cW$ and where the geometrical action on the interval $I$ follows from the arguments in Remark~\ref{remark2}. Moreover, due to the pureness of the vacuum, one gets {\em irreducibility}, i.e., $\gR(\bS^1)= \cB(\gH_0)$.\\
%The proof of the first  identity follows from the validity of the analog
% for associated Weyl algebra and from the fact that von Neumann
%algebra is the weak completion of the GNS representation of the Weyl algebra.
%The proof of irreducibility is quite trivial: First of all notice that $\gR' = \{cI \:|\: c\in \bC\}$. Indeed
%on a hand $\{cI \:|\: c\in \bC\} \subseteq \gR'$, on the other hand   $\gR' \subseteq \pi(\cW)' = \{cI \:|\: c\in \bC\}$
%because  $\pi(\cW)$ is irreducible.
%Taking the commutant once again we find $\gR = \cB(\gH_0)$ because
%$\gR= \gR'' =  \{cI \:|\: c\in \bC\}' = \cB(\cH)$. Notice that it also holds $\gR' \cap \gR = \{cI \:|\: c\in \bC\}$
%and hence $\gR$
%is a factor (the simplest factor of type
%$I_\infty$).\\

\noindent  By all that, the class $\gR$ is termed a {\em local, covariant and irreducible Araki-Haag-Kastler pre-cosheaf}. It is sometimes useful to compare properties having to do with the commutant of the elements of the class $\gR$ with those of the dual pre-cosheaf $\gR^d$, the elements of which are given by
$\gR^d(I)\doteq \gR(I')'$, $I\in\cR$, since in general, by locality, $\gR(I)\subset\gR^d (I)$.  
We want now to prove some further relevant features of this classes of von Neumann algebras. \\

%From now on, if $J\in \cR$ and $\bt = (r,s)\in \bR^2$, the set
%$J+\bt$ is the subset of $\bS^1$ obtained as follows: (1) rotate $J$ by the angle $r$ positively obtaining the set $I$;
%(2) pass to the set $I$ to the set $I_s$ defined in (6) in remark \ref{remark2}. The obtained set is $J+ \bt$ by definition. \\

%\remark \label{remarktransla} If $J\in \cR$ and $\bt=(r,s) \in \bR^2$, the length of $J+\bt = J$
%increases continuously when $|s|$, and
% it rotates positively of an angle $s$.
%$\ell(J+\bt) = \ell(J)$ if and only if $s=0$ in $\bt=(r,s)$.
%As discussed in in (6) in remark \ref{remark2}, $J+\bt \in \cR$ if and only if $\bt=(r,s)$ is such that
% $|s| < \pi -\ell(J)/2$. For $s=\pi -\ell(J)/2$, $J+\bt$ coincides with $\bS^1$ without a point.
%Finally $J+\bt = \bS^1$
%for $|s|> \pi -\ell(J)/2$.\\

\noindent Occasionally, we shall make use of the well-known  possibility \cite{Araki2} that properties of the elements of $\gR$ can be equivalently read at the level of real subspaces of the one-particle Hilbert space. A slight extension of some of the techniques of the seminal paper by Leyland, Roberts and Testard \cite{LRTp} is necessary, although we do not make direct use of Tomita-Takesaki's theory insights \cite{T}. In fact, many of the properties that we shall be concerned with are derived in the literature by use of the so-called Bisognano-Wichmann property (see, e.g., \cite{BGL,BSu}), which relies on the geometrical meaning of the modular operator (usually, either the Lorentz boosts or dilations, depending on case at hand). In our case this is not possible.\\

\noindent First of all notice that if $\psi\in \cH$, the unitary operators 
\beq W[\psi] \doteq e^{\overline{a(\psi) -  a^*(\psi)}}\label{weylgen}\eeq
 are well-defined on the symmetrized Fock space $\gF_+(\cH)$ where $\cH$ is now any complex Hilbert space
 (see \cite{BRII}). 
 These operators satisfy Weyl relations with respect to the
 symplectic form
 \beq \sigma(\psi,\psi') \doteq - 2 \Imm \langle \psi,\psi' \rangle \:, \quad \mbox{for $\psi,\psi' \in \cH$.} 
 \label{Ksigmaext}\eeq
 In the following, if $M \subset \cH$ is a real (not necessarily closed) subspace $M'\subset \cH$ 
 denotes the closed real subspace {\em symplectically orthogonal} to $M$ that is defined as
 $$M' \doteq \left\{ \psi \in \cH\:\: |\:\: \sigma(\psi,\psi')=0 \quad \forall \psi'\in M \right\}\:.$$
 It holds that $M'= \overline{M}' =  \overline{M'}$.
If $M$ is a closed real subspace of $\cH$, the von Neumann algebra 
 generated by all of $W[\psi]$ with $\psi \in M$
 will be indicated by $\gR[M]$.  We shall make use of a fundamental result by Leylard, Roberts and Testard \cite{LRTp}, namely
%\beq
%\gR[M]' = \gR[M'] \:, \quad \mbox{for every real closed subspace $M$ of $\cH$} \label{LRT}\:,
%\eeq
%and
\beq
\gR[M] \cap \gR[N] = \gR(M \cap N)\:, \quad \mbox{for any pair of closed real subspaces $M,N \subset \cH$.}
\label{LRT2}
\eeq
We now specialize to the case where $\cH$ is the one-particle space $L^2(\bS^1,d\theta)$.
If $I\in \cS$ henceforth $M_I \doteq \overline{K(\cS_{I})}$.  
Notice that $M_I \subset M_J$ when $I\subset J$ are elements of $\cR$.
  $\gR[M_I]$ denotes, as we said above,  the  von Neumann algebra generated by operators $W[\psi]$
  with $\psi \in M_I$.
   The symplectic form on $\cH$  defined as in 
 \eqref{Ksigmaext} is an extension 
 of that initially defined on $\cS$  because of \eqref{Ksigma}. Using \eqref{weylgen}
 one build up the unitary operators $W[\psi]$
 with  $\psi \in \cH$. On the other hand, since the $\bR$-linear map $K : \cS \to \cH$
 is injective,  by construction it turns out that $\pi_0(W((\Phi,\Pi))) = W[\psi]$ if $\psi = K(\Phi,\Pi)\in K(\cS)$. 
    Since the $\bR$-linear  map $K(\cS) \ni \psi \mapsto  W[\psi]$ 
 is strongly continuous
 (see for instance \cite{BRII}), we finally obtain that $\gR(I) = \gR[M_I]$.

\subsection{Additivity, Haag duality and regularity}

\noindent As a first observation we indicate that the properties called {\em additivity} and
{\em weak additivity} hold true for the precosheaf $\gR$. 
%In the following,
%if $\{\gA_j\}_{j\in J}$ is a class of sub $*$-algebras  of a given $*$-algebra $\gA$,
%$\vee_{j\in J} \gA_j$ denotes the $*$-algebra {\em finitely generated} by all the algebras $\gA_j$.\\

\begin{proposition}\label{trick} Referring to $\{\gR(I)\}_{I\in \cR}$, if $I \in \cR$ the following properties hold:
\begin{itemize}
\item[$(a)$] {\rm Additivity}: 
if  $\{I_i\}_{i\in L} \subset \cR$ satisfies $\cup_{i\in L} I_i = I$ or respectively $\cup_{i\in L} I_i = \bS^1$, then
$$\left(\bigcup_{i\in L} \bpi_0\left( \cW(I_i)\right) \right)'' = \gR(I) \quad \mbox{and respectively} \quad 
\left(\bigcup_{i\in L} \bpi_0\left( \cW(I_i)\right) \right)'' = \cB(\gH_0)\:.$$
\item[$(b)$] {\rm Weak additivity}:
$$\left(\bigcup_{r\in \bR} \bpi_0\left( \cW(\beta_r I)\right) \right)'' =\cB(\gH_0) \quad \mbox{and} \quad
\overline{\bigcup_{r\in \bR} \bpi_0\left( \cW(\beta_r I)\right)\Psi_0} = \gH_0\:.$$
\end{itemize}
\end{proposition}

\noindent Since the proof does not contain any particularly deep insight into our model but just ordinary construction, we omit it.\\

\noindent Another crucial property we mention is that pioneered by Reeh and Schlieder for local von Neumann algebras \cite{ReehSchlieder}. Its importance is hardly over-emphasized, and it will appear several times in the following. When it holds for the precosheaf $\gR$ we say that it is {\em cyclic}. However, we omit also this proof, since once again it does not offer particular insights into the model, and refer the reader to the vast literature starting from \cite{Araki, Haag}.

\begin{theorem}[Reeh-Schlieder property] \label{RSp} For every $I\in \cR$ the vacuum vector $\Psi_0$ is cyclic for
$\bpi_0(\cW(I))$ and is separating for  $\gR(I)$.
\end{theorem}

\noindent A prominent property of the algebraic approach is Haag duality. It means a form of maximality for the local algebras of observables and is at the basis of most of our treatment. It has been proved several times and in many fashions (see for instance \cite{Osterwalder, Driessler}). When it holds the precosheaf $\gR$ is {\em self-dual}, namely $\gR\equiv \gR^d$. Our aim is to give a self-contained proof in terms of real subspaces, making direct use of properties of dilation operators much as in the paper \cite{LRTp}. However, the non simple connectedness of $\bS^1$ requires some non-trivial variations. The added bonus for our long proof is that it goes first by proving, strategically as common in this affairs, regularity properties of the local algebras.

\begin{theorem}[Haag duality]\label{duality} For every $I\in \cR$ it holds:
$$\gR(I)' = \gR(I')\:.$$
\end{theorem}
\begin{proof}
% By spatial locality one obtains immediately that $\gR(I')\subset \gR(I)'$. It remains to show that
% $\gR(I)' \subset \gR(I')$. To this end we want to use (\ref{LRT}) when $\cH$ is the one-particle space
%of the theory on the circle referred to the vacuum representation $\pi$ we are dealing with.
% In particular $\gF(\cH)= \gH_0$. 
% Therefore the inclusion sufficient to prove Haag duality: $\gR(I)' \subset \gR(I')$,
%  can be re-written $\gR[M_I]' \subset \gR[(M_I)']$. In view of (\ref{LRT})
% to prove Haag duality is enough to establish that \beq (M_I)' \subset M_{I'}\:, \quad
% \mbox{for every $I\in \cR$.} \label{fineH}\eeq
% This is true in view of the lemma subsequent and this proves spatial Haag duality. $\Box$\\
It is enough to show the property at the level of one-particle Hilbert space, i.e.  we should prove that
\beq (M_I)' = M_{I'}\:, \quad
 \mbox{for every $I\in \cR$.} \label{centralnew}\eeq

Now, since
$I$ and $I'$ are disjoint $\sigma((\Phi',\Pi'), (\Phi,\Pi))=0$ if $(\Phi',\Pi') \in \cS_{I'}$ and
$(\Phi,\Pi) \in \cS_{I}$, taking the closures of the space $K(\cS_I)$ and $K(\cS_{I'})$
it must hold
 $(M_I)' \supset M_{I'}$. Therefore to establish the validity of \eqref{centralnew}  is enough to prove the opposite inclusion.
  Our strategy will be the following.
 Take $\psi \in (M_I)'$, we want to show that
 $\psi \in M_{I'}$. The proof of this fact will be decomposed into two proofs corresponding to the following statements:
 
 (I) If $\psi \in (M_I)'$, for every (sufficiently small)
 $\epsilon>0$, $\psi \in M_{I'+ (-\epsilon, \epsilon)}$.

 (II) For every $J\in \cR$,
 \beq
 \bigcap_{\epsilon >0} M_{J + (-\epsilon, \epsilon)} \subset M_J \label{due}\:.
 \eeq

 \noindent {\em Proof of} (I). Since $M_{I'+ (-\epsilon, \epsilon)} = \overline{\cS_{I'+ (-\epsilon, \epsilon)}}$,
  it is sufficient to exhibit
   a sequence $\{(\Phi_k, \Pi_k)\}_{k\in \bN}
 \subset \cS_{I'+ (-\epsilon, \epsilon)}$ such that $K(\Phi_k, \Pi_k) \to \psi$ as $k\to +\infty$.
 Let us prove this fact.
 Looking at \eqref{Ksigma}, we define the distributions $\Phi_\psi,\Pi_\psi \in \mD'(\bS^1)$ individuated by
 \beq
 \int \Phi_\psi(\theta) f(\theta)\: d\theta &\doteq& 2\, \Imm \langle \psi, K(0,f)\rangle\:, \quad \forall f \in C^\infty(\bS^1, \bC)\:, \nonumber\\
 \int \Pi_\psi(\theta) g(\theta)\: d\theta &\doteq& -2\, \Imm \langle \psi, K(g,0)\rangle \:,\quad \forall g \in C^\infty(\bS^1, \bC)\nonumber \:.
 \eeq
 Indeed, using the first definition of \eqref{cfourier} one proves straightforwardly that the linear functionals defined above are continuous in
 the sense of distributions. In the case of  $\Pi_\psi$ one has that the functional individuated -- varying $f$ -- by
 $\langle \psi, A^{-1/4} f\rangle
  = \langle A^{-1/4}\psi, f\rangle$
 is trivially  continuous.
 In the case of $\Phi_\psi$ notice that $\langle \psi, A^{1/4} f \rangle =
 \lim_{n\to +\infty}\langle A^{1/4} \psi_n,  f \rangle$ for some sequence  $\cD(A^{1/4}) \ni  \psi_n \to \psi$ independent from $f$.
 As  each linear functional $\langle A^{1/4} \psi_n,  \cdot \rangle$ is a distribution, $\Phi_\psi$ is a distribution as well.\\
 By construction the distributions $\Phi_\psi$ and $\Pi_\psi$  have supports contained in
 $\overline{I'}$ because, using the definitions, one finds that
 $\int \Phi_\psi(\theta) f(\theta)\: d\theta =0 $ and
 $\int \Pi_\psi(\theta) f(\theta)\: d\theta =0$  for $\psi \in (M_I)'$,
 whenever the smooth function $f$ is supported in $I$. Now consider
 $\rho \in C^\infty(\bR, \bR)$ supported in $(-\epsilon/2, \epsilon/2)$ and define (using weak operatorial
 topology)
 $\psi * \rho \doteq \int_\bR \rho(r) e^{-ir P^\otimes} \psi\:dr \:.$
 Fubini-Tonelli theorem and the fact that $e^{-ir P^\otimes}$ commutes with $A^{\alpha}$ (it can be proved immediately passing in Fourier-series representation) entail that
 $\Phi_{\psi * \rho} = \Phi_\psi * \rho$ and $\Pi_{\psi * \rho} = \Pi_\psi * \rho$,
 where $*$ in the right-hand side denotes the standard convolution so that $\Phi_{\psi * \rho}$ and
  $\Pi_{\psi * \rho}$ are smooth functions supported in $I'+(-\epsilon,\epsilon)$ and thus
  $\psi*\rho \in M_{I'+
  (-\epsilon,\epsilon)}$. Therefore, assuming the existence of a suitable sequence $\{\rho_k\}$ of real smooth functions supported in
  $(-\epsilon/2,
  \epsilon/2)$, with  $\psi * \rho_k \to \psi$,  the sequence of pairs
  $(\Phi_k, \Pi_k) \doteq (\Phi_{\psi * \rho_k},\Phi_{\psi * \rho_k})$
  turns out to be  made of real smooth functions supported in
  $I'+ (-\epsilon, \epsilon)$, 
and $K(\Phi_k, \Pi_k) \to \psi$ holds as requested, proving that $\psi \in M_{I'+ (-\epsilon, \epsilon)}$.\\
To conclude this part,  let us prove the existence of the sequence $\{\rho_k\}$ with $\psi * \rho_k \to \psi$.
  Consider smooth functions $\rho_k \geq 0$ with $\mathrm{supp}\: \rho_k \subset [-1/k,1/k]$ and with $\int_\bR \rho_k(r)
  dr =1$. In our
  hypotheses $||\psi * \rho_k - \psi|| = \left|\left| \int \rho_k(r) e^{-ir P^\otimes} \psi dr  - \psi \right| \right|$ can be re-written as
  \begin{align*}
  \left|\left|  \int_{-1/k}^{1/k}  \left( \rho_k(r) e^{-ir P^\otimes}  - \rho_k(r)\right)  \psi\: dr\right|\right|
  &\leq\int_{-1/k}^{1/k}  \rho_k(r) \left|\left|  \left(  e^{-ir P^\otimes}  - I \right)  \psi \right|
  \right|\: dr \\
  &\leq \sup_{r\in [-1/k,1/k] }\left|\left|  \left(  e^{-ir P^\otimes}  - I \right)  \psi \right|
  \right|\:,
  \end{align*}
   and the last term vanishes as $k\to +\infty$ because $r\mapsto e^{-ir P^\otimes}$ is strongly continuous.
 We have found that $\psi * \rho_k \to \psi$ for $k\to +\infty$ as requested.\\
 
\noindent {\em Proof of}  (II). The proof is based on the following Proposition. (This is a technical point that differentiates our treatment on $\bS^1$  from that in Minkowski space as done in  \cite{LRTp}.)\\
 
  \begin{proposition} \label{proplast}
  Take $J_0\in \cR$ and assume
 $J_0 \equiv (-a,a) \subset (-\pi,\pi] \equiv \bS^1$ with a suitable choice of the origin of $\theta$. There is a class of operators
  $D_\lambda : L^2(\bS^1, d\theta)\to L^2(\bS^1, d\theta)$, with $\lambda$ ranging in a neighborhood $\cO$ of $1$,
 such that, if $\psi \in M_L$ with $\cR \ni L\subsetneq J_0$ :
 \begin{itemize}
  \item[$(a)$] $D_\lambda \psi \in M_{\lambda L}$ and
 
  \item[$(b)$] $D_\lambda \psi \to \psi$ as $\lambda \to 1$.
  \end{itemize}
  \end{proposition}

  \begin{proof} See the appendix A. 
  \end{proof}
 
 \noindent Notice that the requirement $J_0 \equiv (-a,a)$ does not imply any true restriction since all
 the theory is invariant under rotations of the circle.  To go on with the main proof,
 by direct inspection one sees that, for  $\lambda \in (0,1)$,  there is $\epsilon_\lambda>0$ with
  $(c)$ $\lambda \left( J + (-\epsilon_\lambda, \epsilon_\lambda)\right) \subset J$. 
  If $\psi \in  \bigcap_{\epsilon >0} M_{J + (-\epsilon, \epsilon)}$ then  $\psi \in M_{ J + (-\epsilon_\lambda, \epsilon_\lambda)}$
   for every $\lambda \in (0,1)\cap\cO$, so that
 using  $(a)$,  $D_\lambda \psi \in M_{ \lambda(J + (-\epsilon_\lambda, \epsilon_\lambda))}$. Therefore, by $(c)$,
  $D_\lambda \psi \in M_{J}$. Finally, taking the limit as $\lambda \to 1^-$ and using $(b)$ and the fact that $M_J$ is closed, one achieves
  $\psi \in M_J$.
  \end{proof}

 \begin{remark} (${\bf 1}$) Since, by construction  $\bigcap_{\epsilon >0} M_{J + (-\epsilon, \epsilon)} \supset M_J$, the validity of statement
 (II) is in fact equivalent to
 the {\em outer regularity} property:
\beq
 \bigcap_{\epsilon >0} M_{J + (-\epsilon, \epsilon)} = M_J \label{due'}\:.
 \eeq 
By Haag duality and the invariance of $\cR$ under causal complementation, from outer regularity one gets also {\em inner} regularity.\\
(${\bf 2}$) By the properties showed above we say that the precosheaf $\gR$ is {\em local, covariant, irreducible, additive,  cyclic, regular and self-dual}. In the statements of the propositions that will follow, we shall only indicate those properties needed for the proofs.
\end{remark}

\subsection{Definiteness, primarity and punctured Haag duality.}
We pass to prove some other important properties of the class $\gR$. 
%and the analogous class
%  of von Neumann algebras $\{\gR_\lambda(I)\}_{I \in \cR}$ associated with any pure state $\lambda$ with the locally Fock property, i.e. whose GNS-representation is {\bf locally unitarily equivalent} to that of $\omega$, i.e.
%for every $I \in \cR$ there is a unitary  operator $U_I : \gH_0_\lambda \to \gH_0$ with:
%\beq
%U_I \pi_\lambda (a) U^{-1}_I = \pi(a) \:, \quad \mbox{for all $a\in \cW(I)$}\:. \label{locun}
%\eeq
%\begin{remark} Notice that $\gR_\lambda(I)\doteq \pi_\lambda(\cW(I))''$ and,  exploiting  strong operator-topology and
%bijectivity of $U_I$, \eqref{locun}
%implies $U_I\gR_\lambda(I_1)U_I^{-1} = \gR(I_1)$
%and $U_I\gR_\lambda(I_1)'U_I^{-1} = \gR(I_1)'$,
%for every $I_1 \in \cR$ with $I_1 \subset I$ (including the case $I_1=I$).
%\end{remark}
First we shall be concerned with  {\em local definiteness}: It states that
  the algebra of observables associated with a single point $p\in \bS^1$ is the trivial one
  $\bC\cdot\ide$, $\ide$ being the unit element of $\gR$. Since $\{p\} \not \in \cR$, the
  algebra associated with $\{p\}$ is obtained by taking the intersection of the algebras $\gR(I)$ for all
  $I\in \cR$ with $I \ni p$.
  Secondly we shall examine the validity of {\em punctured  Haag duality}, i.e. Haag duality
  seen as on the space $\bS^1 \setminus \{p\}$ for every fixed $p \in \bS^1$.
  Fix $p\in \bS^1$, choose any $I\in\cR$ that does not touch $p$ and define the corresponding $\cI_p \doteq \{ J\in\cR \ |\ J\cap I=\emptyset\ , p\not \in J\ \}$. 
  With this definition, the general statement about the validity of punctured  Haag duality means that:
\beq\label{eq:HD}
\gR(I)\equiv \bpi_0(\cW(I))'' = \bigcap_{J\in\cI_p} (\bpi_0(\cW(J))' \:.
\eeq
  Finally we shall focus on {\em local primarity}. Its validity for $\gR$
  means that each $\gR(I)$ is a factor, hence we shall say that the pre-cosheaf is {\em factorial}.
The proofs are based on the following important result.

\begin{lemma}\label{lemmaadded} If $I,J \in \cR$ and $I \cap J = \emptyset$ then $M_I \cap M_J = \{0\}$ and thus it also holds
$\gR(I) \cap \gR(J) =
\bC\cdot\ide$. In particular, local primarity holds.
\end{lemma}

\begin{proof}
For the first part, see the Appendix A. As far as primarity is concerned, it is trivially true by Haag duality.
\end{proof}

\begin{theorem}\label{Tgenprop} 
The local, additive, and self-dual precosheaf $\gR$
enjoys also the following properties:
\begin{itemize}
\item[$(a)$] Local definiteness;
\item[$(b)$] Punctured Haag duality.
\end{itemize}
Finally,  if $I,J \in \cR$ one has
\beq
\gR(I) \cap \gR(J) = (\gR(K_1) \cup \gR(K_2))'' \label{speriamo2},
\eeq
where  $K_1,K_2 \in \cR$ are the two (possibly empty) components of
$I\cap J$.
\end{theorem}
   
\begin{proof}
$(a)$ We have to show that $\bigcap_{J \in \cR, J\ni p} \gR(J)= \bC\cdot\ide$. This is easily done by using additivity and Haag duality. Indeed, the commutant of the algebra $\bigcap_{J \in \cR, J\ni p} \gR(J)$ is the von Neumann algebra generated by the union over the class of intervals $\{J\in\cR, J\ni p\}$, but this covers $\bS^1$ and by additivity this algebra coincides with $\cB(L^2(\bS^1, d\theta))$. Hence the thesis follows.\\
 %Fix $q\neq p$ with $q\in \bS^1$ and consider $I_1,I_2 \in \cR$ disjoint uniquely determined by
%assuming that $\partial I_1= \partial I_2 = \{p,q\}$.
%If $A \in \bigcap_{J \in \cR, J\ni p} \gR(J)$ it must in particular commute with $\pi(W(\Phi,\Pi))$
%for every choice of $(\Phi,\Pi)$ whose supports are both contained either in $I_1$ or $I_2$.  Therefore
%$A\in \pi(\cW(I_1))' = \pi(\cW(I_1))''' =  \gR(I_1)'$ and
%$ A\in \pi(\cW(I_2))' = \pi(\cW(I_2))''' =  \gR(I_2)'$. We conclude that $A \in \gR[M_{I_1}]' \cap \gR[M_{I_2}]'$.
%That is $A \in \gR[M_{I_2}' \cap M_{I_1}']$, where we employed (\ref{LRT}) and (\ref{LRT2}). Now using
%\eqref{centralnew}, we may assert that $A \in \gR[M_{I_1} \cap M_{I_2}]$ using the fact that $M_{I_2}' = M_{I'_2} = M_{I_1}$
%and $M_{I_1}' = M_{I'_1} = M_{I_2}$. By Lemma \ref{lemmaadded} we conclude that $A \in \gR[\{0\}] = \{cI\}_{c\in \bC}$.\\
% (iv) We have to establish that every $\gR(I)$ is a factor, i.e.
% $\gR(I)' \cap \gR(I) = \{cI\}_{c\in \bC}$.
% The proof of (i) proves, in particular, that each $A \in \gR(I_1)' \cap \gR(I_2)'$ must be of the form $cI$ if
% $I_1$ and $I_2$ are disjoint and fulfil $\partial I_1= \partial I_2 = \{p,q\}$ for arbitrarily fixed $p,q \in \bS^1$.
%This applies in the case
% $I_1 \doteq I \in \cR$, $I_2 = I'$. Therefore we know that $\gR(I)' \cap \gR(I')' = \{cI\}_{c\in \bC}$. However, by
% Haag duality it can be re-written  $\gR(I)' \cap \gR(I) = \{cI\}_{c\in \bC}$, that is $\gR(I)$ is a factor.\\
$(b)$ For punctured Haag duality, if one takes the (triple) commutant of
\eqref{eq:HD}, one gets $\gR(I)' = (\cup_{J\in\cI_p} \gR(J))''$. Hence it is enough, by using Haag duality, to show that it holds
 $$\gR(I')= \left(\gR(I'_1) \cup \gR(I'_2)\right)'' \ ,$$
whenever $I'_1, I'_2 \in \cR$ are disjoint sets, not containing $p$,
such that $I' = I'_1 \cup I'_2 \cup \{p\}$.
 By additivity
$\gR(I') \subset (\gR(I'_1) \cup \gR(I'_2) \cup \gR(J))''$
where $J\in \cR$ is any open set with $J \ni p$. Since it holds for every choice of such $J$, taking the intersection with all such algebras and using $(a)$ one gets
 $$
 \gR(I') \subset \left(\gR(I_1')\cup \gR(I_2')\right)''\ .
 $$
%$\bigcap_{J\in \cR, J\ni p}  \pi(\cW((J)) \subset \bigcap_{J\in \cR, J\ni p}  \gR(J) = \bC\cdot\ide$.
%We have found:
%$\gR(I') \subset \left(\pi(\cW(I'_1)) \cup \pi(\cW(I'_2)) \right)''$.
The other inclusion is trivially true by locality and Haag duality, and the proof is over.\\
To conclude, let us prove (\ref{speriamo2}). We have three possible cases: (1) $K_1=K_2=\emptyset$, (2) $K_1\equiv K\neq\emptyset$, $K_2=\emptyset$, or vice-versa, and finally (3) $K_1\neq\emptyset$, $K_2\neq\emptyset$.\\
Now, case (1) represents Lemma \ref{lemmaadded}. We prove case (2). In case $I\cap J=K$, we have two different possibilities. Either $K$ coincides with $I$ or, similarly, with $J$, or $K$ is a genuine subset of both. In the first case there is nothing to prove. In the second case, we have that $\gR(K)\subset \gR(I)\cap\gR(J)$. However, the opposite inclusion is also true, indeed, 
\beq
\gR(K)\supset \gR(I)\cap\gR(J)\ \Longleftrightarrow\ \gR(K)' \subset (\gR(I)' \cup \gR(J)')''\:,
\eeq
and this last relation is obvious using Haag duality and additivity, when one notices that $I'$ and $J'$ cover $K'$. We now prove case (3). We shall do it by proving first a stronger statement than additivity, being it what is called in the literature {\em strong additivity}. Notice that, if $I \in \cR$ and $I_1,I_2,\ldots,I_n \in \cR$ are pairwise disjoint subsets of $I$ with $\Int\left(\cup_{i=1}^n \overline{I_i}\right) = I$ then
\beq 
\left(\bigcup_{i=1}^n \gR(I_i)\right)'' = \gR(I)\:,\label{speriamo}
\eeq
the proof for $n=2$ is a straightforward consequence of Haag duality and punctured Haag duality used together.
One can iterate the procedure getting the general case with $n$ arbitrary but finite.
Now, coming back to case (3), using \eqref{speriamo} one gets \eqref{speriamo2} by the following procedure.
We first realize that according to the previous decomposition of elements of $\cR$, we can decompose $I$ and $J$ in terms of disjoint intervals as (for instance) $I=\Int \left ( \overline{K_1}\cup \overline{K_2} \cup \overline{J^\prime}\right)$
and $J= \Int \left( \overline{K_1} \cup \overline{K_2} \cup \overline{I^\prime}\right)$. By \eqref{speriamo} and Lemma \ref{lemmaadded} one gets the thesis by taking the intersection of the algebras $\gR(I)$ and $\gR(J)$.  

\end{proof}

\begin{remark}\label{rem:trivial} $({\bf 1})$ The result in \eqref{speriamo2} has been used by Roberts \cite{Roberts1}, 
together with another property (see, \eqref{eq:citp}), as a mean to proof the absence of superselection sectors. It will be used in the following subsection and in the next section.\\
$({\bf 2})$ By strong additivity one easily derives that the pre-cosheaf is $n$-regular for any $n\in\NN$.\\
$({\bf 3})$ Assuming punctured Haag duality, it is possible to derive Haag duality and local definiteness.
\end{remark}

\subsection{Split property and Borchers' Property B}
In this final part we point out some further remarkable properties of the model, namely, the split property and its standard version \cite{DL}, and the Property B \cite{Borchers}.\\
Split property for inclusions of elements of the precosheaf $\gR$ means that, for any given pair of intervals $I_1, I_2\in\cR$ such that the closure of the interval $I_1$ is contained in the interior of the interval $I_2$, there exists an intermediate factor $\gM$ of type I such that $\ \gR(I_1)\subset \gM\subset\gR(I_2)$. The standard split property means that, referring for instance to the situation above, the class satisfies the split property and that the vacuum is cyclic for the algebras $\gR(I_1), \gR(I_2)$, and the relative commutant $\gR(I_1)'\cap\gR(I_2)$.  
Property B of Borchers refers to a feature of the models always present in quantum field theory, namely, that the local algebras are (purely) infinite. This is described by saying that, given an inclusion $\gR(I)\subset\gR(J)$ and a (non-zero and orthogonal) projection operator $E\in\gR(I)$, then $E\sim \ide\ \mathrm{mod}\ \gR(J)$. The condition amounts to say, informally, that in quantum field theory is not possible to lower the degrees of freedom by a local projection.\\

\noindent We have

\begin{lemma} For the local, cyclic and self-dual pre-cosheaf $\gR$ there hold:
\begin{itemize}
\item[$(a)$] Split property;
\item[$(b)$] Standard split property;
\item[$(c)$] Borchers' Property B.
\end{itemize}
\end{lemma}

\begin{proof}
$(a)$ One uses the arguments in \cite{DLR}, that say that split property holds whenever the trace-class condition holds. This means that the ``Boltzmann factor" should be trace-class, i.e.,
$$
\mathrm{Tr}(e^{-\beta H^\otimes})<\infty\ , \qquad \forall\beta>0\ .
$$
By standard arguments (see, e.g., \cite{BRII}), the trace-class condition for the Boltzmann factor for the Hamiltonian in second quantization is implied by the same condition for the single-particle Hamiltonian. In our case, this is obviously true by inspection. Hence split property holds. \\
$(b)$ By the validity of the Reeh-Schlieder property combined with the split property, the inclusions are also standard. Namely, the local algebras $\gR(I_1), \gR(I_2)$ have cyclic and separating vacuum. The relative commutant
$\gR(I_1)'\cap\gR(I_2)$ can be written by Haag duality and property \eqref{speriamo2} as $\gR(I_1')\cap\gR(I_2)= (\gR(K_1)\cup\gR(K_2))''$, and since the algebras $\gR(K_i)$ are cyclic their union is such.\\
$(c)$ This part follows by $(b)$ and arguments given, for instance, by Roberts (Corollary 10.2) in \cite{Roberts}.
\end{proof}

\begin{remark}\label{rem:canonical} $({\bf 1})$ Property B implies that the local algebras are factors of type III. We do not investigate whether they are even type $\mathrm{III}_1$.\\
$({\bf 2})$ As a matter of fact the property of inclusions being split yields that if the triple $(\gR(I),\gR(J),\Psi_0)$ is standard split, then there exists a normal and faithful product state on the von Neumann algebra generated by $\gR(I)$ and $\gR(J)'=\gR(J')$. This entails that the last algebra is {\em canonically} isomorphic to the (von Neumann) tensor product, i.e. 
\beq\label{eq:citp}
(\gR(I)\cup\gR(J'))'' \approx \gR(I)\otimes\gR(J')\ .
\eeq
$({\bf 3})$ So finally, collecting together all properties,  our precosheaf $\gR$ is {\em local, irreducible, additive, cyclic, regular, self-dual, definite, factorial and split}.
\end{remark}

\section{Superselection sectors} \label{sec:sss}
Following \cite{BR2}, in this section we show how  the nontrivial topology
of $\bS^1$ individuates nontrivial superselection sectors. As a fact, that will be proved shortly, sectors of Doplicher, Haag and Roberts kind are absent, besides that of the vacuum. Nonetheless, one can construct explicit examples of charges of topological origin.

%The found situation is however puzzling here, because some of the most important 
%results presented in the mentioned paper (which are valid for dimension of the spacetime $\geq 3$) cannot be used here and, as matter of fact, they do not hold true. In particular the theorem of localization of fundamental group, theorem 4.1 in \cite{BR2}, fails to be valid
%specializing to our definition of cocycles as it can be verified by direct inspection. Therefore, some relevant consequences 
%of that theorem may be false in our low-dimension context, and the extent has to be analyzed carefully.
%In particular, the absence of  irreducible cocycles different from characters when the fundamental group of the manifold is Abelian, as established in  Corollary 6.8 in \cite{BR2}, no longer holds. Indeed, the irreducible cocycles we are going to 
% present below are a clear counterexample. Furthermore the proof of Theorem 4.3 
%in \cite{BR2}, which demonstrates 
% equivalence of the category of $1$-cocycles and the category of the representations of the algebra of observables satisfying the (topological) selection criterion, does not hold since it relies upon theorem 4.1 in \cite{BR2}.
% On the other hand, as a remarkable result which is still valid in our low-dimension context,
% following the way outlined in the proof of that theorem, the cocycles we present shortly 
%individuate generalized representations of the algebra of observables satisfying the (topological) 
%selection criterion defined in  \cite{BR2}.
% 

\subsection{Generalized representations and cohomology}
>From now on we adopt definitions and conventions in \cite{RuzziRMP} and  \cite{BR2} specialized
to our case, concerning
net cohomology of posets on $\bS^1$. The employed poset will be $\cR$ equipped with the partial ordering relation $\subseteq$.
Our {\em reference net of  observables}
is the net 
of
unital $C^*$-algebras  $$\cW : \cR \ni I \to \cW(I)\:.$$
If $\widetilde{I} \subseteq I$, the natural isometric  $*$-homomorphisms given by inclusion maps of $\cW\left(\widetilde{I}\right)$ into 
$\cW(I)$ will be denoted by $j_{I\widetilde{I}}$ and are the {\em inclusion morphisms} as in \cite{BR2}. The coherence requirement $j_{I'I} = j_{I'\tilde{I}} j_{\tilde{I}I}$ for ${I}\subseteq \widetilde{I} \subseteq I'$ is trivially fulfilled. \\

\noindent{\em Generalized representations}.   A {\em unitary generalized representation} on $\gH_0$ (for the pre-cosheaf $\cW$) in the sense of \cite{BR2}
is a pair $\{\pi,\psi\}$, where $\pi$ denotes a function that associates
a representation $\pi_I$ of  $\cW(I)$ on the fixed common Hilbert space $\gH_0$ 
 with any $I \in \cR$; $\psi$ denotes a function
that associates a unitary linear operator $\psi_{I\tilde{I}}\in\gB(\gH_0)$
 with any pair $I,\widetilde{I}\in \cR$, with
$\widetilde{I}\subseteq I$. The functions $\pi$ and $\psi$ are required to satisfy the following relations
\beq
\psi_{I\tilde{I}} \pi_{\tilde{I}}(A) = \pi_I j_{I\tilde{I}}(A) \psi_{I\tilde{I}}\:, \quad
A \in \cW\left(\widetilde{I}\right)\:, \: \widetilde{I} \subseteq I\:, \quad  \mbox{and}
\quad \quad  \psi_{I'{I}} \psi_{I\tilde{I}} = \psi_{I'\tilde{I}} \:, \quad
\widetilde{I} \subseteq I \subseteq I' \:.\label{UNR}
\eeq

\remark {\em A priori}, different Hilbert spaces can be used for each representation $\pi_I$ \cite{BR2}. However, we are interested here to base all generalized representations on the GNS Hilbert space of the vacuum, so we make this choice just from the beginning. Similarly the unitarity requirement on the operators $\psi_{I\tilde{I}}$ may be dropped (see \cite{BR2} for the general case).\\

\noindent An {\em intertwiner} from $\{\pi,\psi\}$ to $\{\rho, \phi\}$ is a function $T$ associating a bounded operator
$T_I \in\gB(\gH_0)$ with any $I \in \cR$, and satisfying the relations
\beq
T_I \pi_{I} = \rho_I T_I\:, \quad \quad \mbox{and}
\quad \quad \quad T_I \psi_{I\tilde{I}}  = \phi_{I\tilde{I}}  T_I \:, \quad
\widetilde{I} \subseteq I  \:. \label{intertREP}
\eeq
We denote the set of intertwiners from $\{\pi,\psi\}$ to $\{\rho, \phi\}$ by the symbol  $(\{\pi,\psi\},\{\rho, \phi\})$,
and say that the net representations are {\em unitarily equivalent} if they have a unitary
intertwiner $T$, that is, $T_I$ is a unitary operator for any $I\in \cR$.
$\{\pi,\psi\}$ is {\em irreducible} when the unitary elements of $(\{\pi,\psi\},\{\pi, \psi\})$  are of the form $c\ide$ with 
$c\in \bC$ and $|c|=1$.
Motivations for the given definitions can be found in \cite{BR2} and in the literature quoted therein.\\

\noindent {\em Simplices and cocycles}.  Let us pass to introduce $1$-cocycles of $\cB(\gH_0)$.
In the following $\Sigma_k(\cR)$ will denote the class of singular $k$-simplices of $\cR$ (with $\Sigma_0(\cR)= \cR$).
The inclusion maps $d_i^n: \Delta_{n-1} \to \Delta_n$ between standard simplices -- where,
 following \cite{RuzziRMP}, $\Delta_n$ is the standard $n$-simplex --
are extended to maps $\partial^n_i: \Sigma_n(\cR) \to \Sigma_{n-1}(\cR)$, 
called {\em boundaries}, between singular simplices by 
setting
$\partial_i^n f \doteq f\circ d^n_i$, where $f: \Delta_n \to \cR$ is the order preserving map defining the 
singular $n$-simplex
 of $\cR$. One can easily check, by the definition of $d_i^n$ \cite{RuzziRMP}, that the following relations hold:
 $$\partial_i^{n-1}\circ \partial_j^n = \partial_j^{n-1}\circ \partial_{i+1}^n\:, \quad i\geq j\:.$$
 From now on, we will omit the superscripts from the symbol $\partial_i^n$.
A {\em path} $p$ from $I_0 \in \cR$ to $I_1\in \cR$ is an ordered set $\{b_1,b_2,\ldots,b_n\} \subset \Sigma_1(\cR)$ (where $n$
is an arbitrarily fixed integer depending on the path)
such that $\partial_0b_1=I_0$, $\partial_1b_n=I_1$ and $\partial_0 b_k = \partial_1 b_{k-1}$, for the remaining cases.
$P(I_0,I_1)$ denotes  the class of paths from $I_0 \in \cR$ to $I_1\in \cR$.
  $\pi_1(\cR)$ will indicate the {\em fundamental group} of $\cR$ as established in Definition 2.4 of \cite{RuzziRMP}
  making use of the above-defined notion of paths of  $P(I_0,I_0)$ for some fixed basepoint $I_0 \in \cR$,
   taking patwise-connection of $\cR$ into account. Its definition follows straightforwardly from the analogous definition
 based on the notion of continuous path in a topological space. In particular it arises that $\pi_1(\cR)$ does not depend on the basepoint $I_0$.
   In view of Theorem 2.18 in \cite{RuzziRMP}, since $\bS^1$ is Hausdorff, arcwise connected and $\cR$ is a topological base of
$\bS^1$, it turns out that $\pi_1(\cR)$ coincides with the fundamental group of $\bS^1$, i.e. $\pi_1(\cR) = \bZ$ in our case.
 Finally, it is worth
remarking that every irreducible unitary representation of $\bZ$, $\{\lambda_x(n)\}_{n\in\bZ}$ is one-dimensional, $\lambda_x(n) : \bC \to \bC$,
as the group is Abelian. All those representations are one-to-one labeled by $x\in \bR$ and have the form:
\beq \lambda_x(n) : \bC \ni \alpha \mapsto e^{i n x} \alpha\:,\quad\mbox{for all $n \in \bZ$}. \label{lambdax}\eeq
A $1$-cocycle in $\cB(\gH_0)$
is a field $z : \Sigma_1(\cR) \ni b \mapsto z(b) \in \cB(\gH_0)$ of unitary operators satisfying the {\em $1$-cocycle identity}:
\beq
z(\partial_0 c) z(\partial_2c) = z(\partial_1c)\:, \quad \mbox{for all $c\in \Sigma_2(\cR)$.} \label{one-cocycle}
\eeq
A $1$-cocycle $z$ is said to be a {\em coboundary} if it can be written as $z(b)= W^*_{\partial_0b } W_{\partial_1b }$, $b \in \Sigma_1(\cR)$, for some field of unitaries $\cR \ni I \mapsto W(I) \in \cB(\gH_0)$.
The space of $1$-cocycles will be indicated by $Z^1(\cR, \cB(\gH_0))$.
 Following \cite{RuzziRMP} we say that 
$z, z_1 \in Z^1(\cR, \cB(\gH_0))$ are {\em equivalent in $\cB(\gH_0)$} if  they admits a unitary {\em intertwiner}, i.e.  a 
field $V: \cR \ni I \mapsto V_I \in \cB(\gH_0)$
of unitary operators such that
\beq
V_{\partial_0 b} z(b) = z_1(b) V_{\partial_1 b}\:, \quad  \quad \mbox{for all $b\in \Sigma_1(\cR)$.} \label{equivalence}
\eeq
The set of fields $V$ satisfying (\ref{equivalence}) is denote by $(z,z_1)$.
A $1$-cocycle is said to be {\em trivial} if it is equivalent in $\cB(\gH_0)$ to the  cocycle
$z :  \Sigma_1(\cR) \ni b \mapsto I$, and this is equivalent to say that $z$ is a coboundary.
A $1$-cocycle $z \in Z^1(\cR, \cB(\gH_0))$ is said to be {\em irreducible} if
there are no non-trivial unitary intertwiners  in $(z,z)$. \\ 
The family whose objects are cocycles and whose arrows associated with cocycles  $z$ and $z_1$
are the intertwiners of $(z,z_1)$ forms a category denoted by $Z^1(\cR, \cB(\gH_0))$.\\
Given a unitary generalized representation $\{\pi,\psi\}$ of $\cW$ over $\gH_0$ define
\beq 
\zeta^\pi(b) \doteq \psi^*_{|b|,\partial_0b} \psi_{|b|,\partial_1b}\:, \quad b\in \Sigma_1(\cR)\:,\label{zetapi}
\eeq
as usual $|b|\in \cR$ denotes the support of the simplex $b$.
One can check that $\zeta^\pi$ is a $1$-cocycle of $Z^1(\cR, \cB(\gH_0))$. 
$\{\pi,\psi\}$ is said to be {\em topologically trivial} if $\zeta^\pi$ is trivial.\\
It can be proven \cite{BR2} that
 if the unitary generalized representations $\{\pi,\psi\}$ and $\{\rho,\phi\}$ are unitarily equivalent, then the corresponding
$1$-cocycles $\zeta^\pi$ and $\zeta^\phi$ are equivalent in $\cB(\gH_0)$; moreover, if
the unitary generalized representation $\{\pi,\psi\}$ is topologically trivial, then it is equivalent to one of the form $\{\rho,\ide\}$, where  all
 $\ide_{\tilde{I} I}$ are  the identity operators.\\
Finally we remind that
the relation between $Z^1(\cR, \cB(\gH_0))$ and $\pi_1(\gR)$ is obtained as follows (Theorem 2.8 in \cite{RuzziRMP} specialized to the case of $\cR$).  If $p=\{b_1,\ldots,b_n\} \in P(I_0,I_1)$, one defines
$z(p)\doteq z(b_n)z(b_{n-1})\cdots z(b_1)$.  \\

\begin{theorem} \label{teoruzzi} Consider $z\in Z^1(\cR, \cB(\gH_0))$ and
fix a $I_0\in \cR$. For every path $p\in P(I_0,I_0)$ and the associated element $[p] \in \pi_1(\gR)$, define
\beq
\pi_z([p]) \doteq z(p)\:, \label{piz}
\eeq
The map $Z^1(\cR, \cB(\gH_0)) \ni z \mapsto \pi_z$ is well defined and maps $1$-cocycles $z$ to unitary
representations $\pi_z$ of $\pi_1(\cR)$ in $\gH_0$.
If $z,z_1 \in Z^1(\cR, \cB(\gH_0))$ are equivalent in $\cB(\gH_0)$ the corresponding representations $\pi_z, \pi_{z_1}$ of $\pi_1(\cR)$ are unitarily
equivalent. Finally, up to equivalence, the map $Z^1(\cR, \cB(\gH_0)) \ni z \mapsto \pi_z$ is injective.
\end{theorem}

\noindent Notice that, as a consequence, $z\in Z^1(\cR, \cB(\gH_0))$ is trivial if and only if the associated representation of
$\pi_z$ is the trivial.\\

\noindent {\em Topological superselection sectors}.  
 Let us pass to the selection criterion and the topological superselection sectors introduced in 
\cite{BR2}.
 Consider the unitary generalized representation $({\bpi_0}, \ide)$ of $\cW$ over $\gH_0$. It enjoys the following properties: It is faithful and defined over a (complex infinite dimensional) separable Hilbert space $\gH_0$ and we have seen that the pre-cosheaf of von Neumann algebras $\gR$ is irreducible, cyclic, self-dual, regular and split. 
Finally $({\bpi_0}, \ide)$  is topologically trivial since the cocycle $\zeta^{\bpi_0}$ associated with $({\bpi_0}, \ide)$ is the simplest co-boundary (hence the associated unitary representation of $\pi_1(\bS^1)$ is trivial, too).
All these requirement are those assumed in \cite{BR2} to define a {\em reference representation} and state 
a selection criterion which generalizes DHR criterion.\\
Following \cite{BR2} we say that a unitary generalized representation $\{\pi,\psi\}$ over $\gH_0$
is {\em a sharp excitation of the reference representation} $\{{\bpi_0}, \ide\}$, if for any $O\in \cR$ and for any simply connected open set $N \subset \bS^1$, such that $\overline{O} \subset N$, there holds
\beq \{\pi,\psi\}\rest_{O'\cap N} \quad \cong  \quad \{{\bpi_0}, \ide\}\rest_{O'\cap N}\:. \label{SS}\eeq
This amounts to saying that there is a family $W^{NO} 
\doteq \{W^{NO}_I \:|\:  \overline{I} \subset N\:, I \subset  O'\}$ of 
unitary operators in $\gH_0$ such that
\begin{itemize}
\item[$(1)$]   $W^{NO}_I\pi_I = {\bpi_0}_I W^{NO}_I$;
\item[$(2)$]    $W^{NO}_I\psi_{I\widetilde{I}} = W^{NO}_{\widetilde{I}}$, for all $\widetilde{I} \subset I$;
\item[$(3)$]   $W^{NO}  = W^{N_1O}$  for any simply connected open set  $N_1$ with $N \subset  N_1$.
\end{itemize}

\noindent These three requirements represent the {\em selection criterion}.
It turns out that $W^{NO}$ is independent form the region $N$. 
The unitary equivalence classes of irreducible unitary generalized representations satisfying the selection criterion are the 
{\em superselection sectors} and the analysis of their charge structure and
topological content, in the case of a generic globally hyperbolic spacetime  with dimension $\geq 3$
was the scope of the work \cite{BR2}. We are dealing with a (particular) $2$-dimensional spacetime, so we expect that some of the results found there cannot apply.\\

\noindent{\em Localized cocycles}. One of the most important result, established in Theorem 4.3 \cite{BR2}, is that, for globally hyperbolic spacetimes with dimension $\geq 3$, the $C^*$-category whose objects 
are {\em sharp excitations} of $\{{\bpi_0}, \ide\}$, with arrows given by 
intertwiners (\ref{intertREP}), is equivalent to the subcategory of $1$-cocycles whose objects and arrows fulfill 
a natural  localization property, as far as applications to quantum field theory are concerned, with respect to ${\bpi_0}$. Let us define this category specializing to our case\footnote{Where, however, there is no guarantee for the validity of the equivalence theorem.}. \\
We define now the category of (localized) cocycles $Z^1(\gR)$ with respect to the reference representation $\{{\bpi_0}, \ide\}$
and the associated pre-cosheaf of von Neumann algebras $\gR$. The objects of $Z^1(\gR)$ are 
$1$-cocycles $z\in Z^1(\cR,\cB(\gH_0))$ fulfilling the further localization requirement 
\beq
z(b) \in  \gR(|b|)\ ,
\eeq
 for every $b\in \Sigma_1(\cR)$, and whose arrows are the unitary intertwiners $V$ fulfilling the analogous  localization requirement 
 \beq
 V_I \in \gR(I)\ ,
 \eeq
for every $I \in \cR$.
In particular, we say that two cocycles
$z, z_1 \in Z^1(\cR)$ are {\em equivalent} if they are equivalent in $\cB(\gH_0)$
by means of a unitary intertwiner which satisfies the localization requirement mentioned before.
A $1$-cocycle $z \in Z^1(\gR)$ is said to be {\em irreducible} if there are no unitary intertwiners  in $(z,z)$ satisfying
 the localization requirement different from $c\ide$ for $c\in \bC$ with $|c|=1$. Notice that irreducibility in $Z^1(\cR,\cB(\gH_0))$ is much stronger than irreducibility 
 in $Z^1(\gR)$. \\ From now on we consider cocycles in $Z^1(\gR)$ only.
 In the following we establish the existence of nontrivial elements of $Z^1(\gR)$. Afterwards we show that every 
 such $1$-cocycle individuates a class of (unitarily equivalent) generalized representations of $\gR$ verifying 
the selection criterion. 
 
\remark The absence of  irreducible cocycles different from characters when the fundamental group of the manifold is Abelian, as established in  Corollary 6.8 in \cite{BR2}, no longer holds in low dimension. 
Indeed, our whole business, later on, will be on showing explicit examples.

\subsection{Triviality of $Z^1_t(\gR)$}
We wish to deal with the intuitive idea that in our model there are no sectors of DHR type besides that of the vacuum. 
This is based on the fact, proven in \cite{BR2}, that the new selection criterion includes the one of Doplicher, Haag and Roberts. Indeed, if we consider representations of the form $(\bpi, \ide)$ satisfying 
\eqref{SS}, it is easy to show that the intertwiners $W$ between the representations $\bpi$ and $\bpi_0$, do not depend on the choice of the regions $N$ and $I$, as shown in the requirements $(1)$,$(2)$ and $(3)$, following \eqref{SS}. Hence, we are back the criterion originally introduced by the cited authors, i.e. that $\pi$ is {\em locally} unitarily equivalent to the vacuum representation $\bpi_0$, 
\beq\label{DHRC}
\pi\rest_{O'}\  \cong\   \bpi_0 \rest_{O'}\ ,\qquad O\in\cR.
\eeq

One way to prove that in our model this entails that the two representations are {\em globally} unitarily equivalent, would be to show that the first cohomology associated with representations of the form $(\bpi,\ide)$, satisfying \eqref{DHRC}, is (quasi-) {\em trivial}, following the germinal idea of Roberts (see, e.g., \cite{Roberts1,Roberts}). The subcategory of $Z^1(\gR)$ composed by objects as 1-cocycles associated with such representations is 
termed $Z_t^1(\gR)$. As recalled in the introduction, we wish to emphasize that M\"uger \cite{MM}, respectively Ciolli \cite{FC}, had independently used similar ideas to prove the absence of non trivial sectors in the case of general massive scalar field theories, respectively free massless scalar field theories, on two dimensional Minkowski spacetime.\\

We cite for record the criterion of Roberts, rephrased for our purposes. If $p=\{b_1,\dots,b_n\}$ is any path we define $\partial b$ to be the ordered set $\{\partial_0 b_1,\partial_1 b_n\}$, and we identify any 1-simplex $b$ with a path 
$\{b\}$. One has;

\begin{theorem} Let $\gR$ be a pre-cosheaf of von Neumann algebras over $\gR$ satisfying for each $b\in\Sigma_1(\gR)$ the following conditions:
\begin{itemize}
\item[$(a)$] $\displaystyle \bigcap_{\partial p=\partial b} \gR(|p|) = (\gR(\partial_0 b)\cup\gR(\partial_1 b))''\ ,$
\item[$(b)$] If $\overline{\partial_0 b} \subset (\partial_1 b)'$, then the von Neumann algebra generated by $\gR(\partial_0 b)$ and $\gR(\partial_1 b)$ is canonically isomorphic to the von Neumann tensor product $\gR(\partial_0 b)\otimes\gR(\partial_1 b)$.
\end{itemize} 
Then, for any $z\in\ Z^1_t (\gR)$ there are associated unique Hilbert spaces with support $\ide$, $H(I)\in\gR(I)$, $I\in\cR$, such that $z(b)H(\partial_1 b) = H(\partial_0 b)$, $b\in\Sigma_1(\gR)$. In particular, any object of $Z^1_t(\gR)$, is a direct sum of trivial 1-cocycles.
\end{theorem}

\noindent We shall focus on proving the hypothesis of the Theorem, and address the reader to the cited literature for the understanding of the many details connected with its statement. \\

We notice immediately that the conditions $(a)$ and $(b)$ have been already verified for our {\em local, irreducible, additive, cyclic, regular, self-dual, definite, factorial and split} pre-cosheaf $\gR$. Indeed, as far as the first condition $(a)$ is concerned, it suffices to prove that it holds $ \gR(\partial_0 b)\cap\gR(\partial_1 b) = (\gR(\partial_0 b)\cup\gR(\partial_1 b))'' $, which however we recall to be condition \eqref{speriamo2}, proven in Theorem~\ref{Tgenprop}; whilst the second one is the property \eqref{eq:citp}, coming from the split property.\\

\noindent The above proves that any 1-cocycle corresponding to the representation $\pi$ (non necessarily irreducible), satisfying the DHR criterion is either a trivial, or a direct sum of trivial 1-cocycles.

 \subsection{Existence and properties of nontrivial elements of $Z^1(\gR)$.}  Let us construct localized $1$-cocycles w.r.t. the reference representation $\{{\bpi_0}, \ide\}$. \\
To this end we need a preliminary construction. 
First of all, let us fix an orientation (for instance, anti-clockwise) to $\bS^1$ and, in the following, we shall refer to that orientation for assigning the initial and final points
to $0$-simplices.
Afterwards, assign smooth functions to $0$-simplices
\beq  
\chi: \cR \ni I \mapsto \chi^{I} \in C^\infty(I, \bR) \label{assignchi}
\eeq
such that, referring to the chosen orientation of $\bS^1$:
\begin{itemize}
\item[$(i)$]  $\chi^{I}(\theta) \in [0,1]$,
\item[$(ii)$] $\chi^{I}(\theta) =0$ in a  neighborhood of the initial point of $I$,
\item[$(iii)$] $\chi^{I}(\theta) =1$ in a  neighborhood  of the final point  of $I$. 
\end{itemize}
Now, consider a $1$-simplex $b$. Extend
$\chi^{(\partial_1b)}$ and $\chi^{(\partial_0b)}$ smoothly and uniquely as constant functions over $|b| \setminus \partial_1b$ and $|b| \setminus \partial_0b$. The functions so extended over the whole $|b|$ will be denoted by $\chi^{(\partial_1b)}$ and $\chi^{(\partial_0b)}$ again. Finally, for every $b\in \Sigma_1(\cR)$ we define the function $\chi^{(b)} \in C^\infty(|b|, \bR)$
\beq \chi^{(b)} \doteq \chi^{(\partial_1b)} -  \chi^{(\partial_0b)} \label{assignchi2}\:.\eeq
notice that this function vanishes in a neighborhood of each endpoint of $|b|$. Therefore 
$\chi^{(b)}$ can be extended uniquely to a smooth function defined on the whole circle $\bS^1$ and supported in $|b|$. We shall
denote by $\chi^{(b)}$ again this unique extension.\\
Let us come to $1$-cocycles. 
We define (changing slightly notation) 
\beq
Z(f,g)\doteq {\bpi_0}\left(W(f,g) \right) =  W\left[2^{-1/2} (A^{1/4} f + iA^{-1/4} g)\right]\:, \quad \mbox{for $(f,g) \in \cS$}\:. \label{Vf}
\eeq

In the following, to define a $1$-cocycle localized at $b\in \Sigma_1(\cR)$, we shall replace the arguments $f$ 
and $g$ with $\chi^{(b)}$-smeared restrictions of those functions to 
$0$-simplices $|b|$ for any $1$-simplex $b$. The restriction is necessary in order to fulfill the localization requirement of $1$-cocycles.
 The smearing procedure is necessary too, at least for the entry of $A^{1/4}$, whose domain generally
does not includes elements $\chi_{|b|}g$,
$\chi_{|b|}$ being the characteristic function of the set $|b|$.
 It, however, includes  every smoothed function
 $\chi^{(b)} g$ when $(f,g) \in \cS$. \\
We are now in place to state our first result, showing the existence of topological $1$-cocycles. The following theorem also establishes 
the independence from 
$\chi$, up to equivalence, the irreducibility of cocycles and the fact that they are inequivalent 
if $f\neq f'$ or $g\neq g'$. \\

\begin{theorem} \label{propcocycle}
Fix an orientation of $\bS^1$, an assignment   $\chi : \cR \ni I \mapsto \chi^{I}$ as
in (\ref{assignchi}) and define $\chi^{(b)}$ as in (\ref{assignchi2}).
For every choice of $(f,g) \in \cS$ the map 
\beq
z^{(\chi)}_{(f,g)} : \Sigma_1(\cR) \ni b \mapsto Z\left( \chi^{(b)}f, \chi^{(b)}g\right) \label{cocyclef}\:,
\eeq
 is a  $1$-cocycle of $Z^1(\gR)$. The following further facts hold.
 \begin{itemize}
\item[$(a)$] Every $1$-cocycle $z^{(\chi)}_{(f,g)}$ is irreducible.
\item[$(b)$]  For fixed $(f,g)\in \cS$, but different assignments $\chi_1,\chi_2$, 
 $z^{(\chi_1)}_{(f,g)}$ and $z^{(\chi_2)}_{(f,g)}$ are equivalent.
\item[$(c)$]   For a fixed assignment $\chi$,  $z^{(\chi)}_{(f,g)}$ and $z^{(\chi)}_{(f',g')}$ are equivalent if and only 
 if $f=f'$ and $g=g'$.
\item[$(d)$] If the assignment $\chi : \cR \ni I \mapsto \chi^{I}$ is covariant\footnote{Covariant assignments $\chi : \cR \ni I \mapsto \chi^{I}$
with respect to the isometry group of $\theta$-displacement on $\bS^1$ do exist
 as the reader can easily prove.}
with respect to the isometry group of $\theta$-displacement on $\bS^1$:
\beq\chi^{(\beta_r(I))} = \beta^*_r(\chi^{(I)}) \quad \forall I \in \cR, \forall r\in \bR, \label{covset} \eeq
then, for every $(f,g) \in \cS$  and for every $r\in\bR$ and $b\in \Sigma_1(\cR)$,
$$U_{(r,0)} z^{(\chi)}_{(f,g)}(b) U_{(r,0)}^* = z_{(\beta^*_r (f) ,\beta^*_r(g))} (\beta_r(b))\:,$$
where $U_{(r,0)}$ is the one-parameter unitary group implementing $\theta$-displacements $\beta_r$ and leaving the vacuum invariant, introduced in 
theorem \ref{thoremvacuum1} and $\beta_r^*$ is the pull-back action of $\theta$-displacements on functions defined on $\bS^1$.
\end{itemize}
\end{theorem}

\begin{proof}
Let us prove that \eqref{cocyclef}  defines a localized $1$-cocycle. First we notice that the unitary operator
$$Z(\chi^{(b)}f,\chi^{(b)}g)\doteq W\left[2^{-1/2} (A^{1/4} \chi^{(b)}f + iA^{-1/4} \chi^{(b)}g)\right]$$
is an element of $\gR(|b|)$ since $\supp (\chi^{(b)}) \subset |b|$ as noticed previously.
 So, the identity \eqref{one-cocycle} remains to be proved. Let us consider a 
 $2$-simplex $c$. To simplify the notation we define $b_k \doteq \partial_k c$ for $k=0,1,2$.
 Since $|c|$ cannot coincide with the whole circle (and this is the crucial point),
 all functions $\chi^{(\partial_1 b_j)}$ can be extended, uniquely and smoothly, to functions defined on $|c|$ as constant functions outside their original domain. The extension procedure does not affect the definition of the functions $\chi^{(b_i)}$. We shall exploit this extension from now on.
 We have to show that $z^{(\chi)}_{(f,g)}(b_0)z^{(\chi)}_{(f,g)}(b_2) = z^{(\chi)}_{(f,g)}(b_1)$, that is
 $$Z\left(\chi^{(b_0)}f, \chi^{(b_0)} g\right)
Z\left(\chi^{(b_2)}f, \chi^{(b_2)} g\right) = Z\left(\chi^{(b_1)}f, \chi^{(b_1)} g\right)$$
that is, in turn, 
 $$Z\left((\chi^{(\partial_1b_0)}-\chi^{(\partial_0b_0)})f,(\chi^{(\partial_1b_0)}-\chi^{(\partial_0b_0)})g\right)
Z\left(
(\chi^{(\partial_1b_2)}-\chi^{(\partial_0b_2)})f,(\chi^{(\partial_1b_2)}-\chi^{(\partial_0b_2)})g\right)$$
\beq  = Z\left((\chi^{(\partial_1b_1)}-\chi^{(\partial_0b_1)})f,(\chi^{(\partial_1b_1)}-\chi^{(\partial_0b_1)})g\right)\:.\label{idcocycle}\eeq
Now notice that, in view of the definition of a $2$-simplex, $\partial_1 b_1= \partial_1b_2$, 
$\partial_0 b_0= \partial_0b_1$ and $\partial_0 b_2 = \partial_1 b_0$, so that the left-hand side of (\ref{idcocycle})
can be rewritten as
$$ Z\left((\chi^{(\partial_1b_1)}-\chi^{(\partial_0b_2)})f,(\chi^{(\partial_1b_1)}-\chi^{(\partial_0b_2)})g\right)
Z\left(
(\chi^{(\partial_0b_2)}-\chi^{(\partial_0b_1)})f,(\chi^{(\partial_0b_2)}-\chi^{(\partial_0b_1)})g\right)\:,$$
where {\em all} functions $\chi^{(\partial_ib_j)}$ are now defined on the whole $|c|$
and the differences $\chi^{(\partial_0b_2)}-\chi^{(\partial_0b_1)}$, $\chi^{(\partial_1b_1)}-\chi^{(\partial_0b_2)}$ are  defined everywhere on $\bS^1$ and compactly supported in $|c|$. Finally, making use of Weyl relations, taking the definition (\ref{Vf}) of $Z(f,g)$ into account, we find that the terms $\pm\chi^{(\partial_0b_2)}$ cancel each other in the final exponent, and the left-hand side of (\ref{idcocycle}) is:
$$Z\left((\chi^{(\partial_1b_1)}-\chi^{(\partial_0b_1)})f,(\chi^{(\partial_1b_1)}-\chi^{(\partial_0b_1)})g\right)
e^{i \int_{\bS^1} (\chi^{(\partial_1b_1)}-\chi^{(\partial_0b_2)})(\chi^{(\partial_0b_2)}-\chi^{(\partial_0b_1)})(fg-gf)d\theta} \:.$$
Since the phase vanishes, we have found the right-hand side of  \eqref{idcocycle}.\\
$(a)$ Let us pass to the irreducibility property of the defined cocycles.
Let $V: \cR \ni I \mapsto V_I \in \gR(I)$ be a field of
 unitary operators such that
\beq
V_{\partial_0 b}  = z^{(\chi)}_{(f,g)}(b) V_{\partial_1 b} z^{(\chi)*}_{(f,g)}(b)\:, \quad  \quad \mbox{for all $b\in \Sigma_1(\cR)$.} \label{equivalence3}
\eeq
Since $V_{\partial_1 b} \in \gR(\partial_1b)$, then $V_{\partial_1 b} = \sum_i c_i\pi( W(r_i,s_i))$ where $r_i,s_i$ are smooth real 
functions 
supported in $\partial_1b$, $c_i \in \bC$, and the series converges in the strong operatorial topology. Therefore, using Weyl relations, and the continuity of multiplications for the strong operatorial topology,
$$ z^{(\chi)}_{(f,g)}(b) V_{\partial_1 b} z^{(\chi)*}_{(f,g)}(b) = \sum_k c_k z^{(\chi)}_{(f,g)}(b)\pi(W(r_k,s_k)) z^{(\chi)}_{(f,g)}(b) 
= \sum_k c_k \pi(W(r_k,s_k)) \exp\{i\vphi_k\}\:,$$
for some $\vphi_k\in\bR$. The final series converges in the strong operatorial topology, too.
Since $\pi(W(r_k,s_k)) \in \gR(\partial_1 b)$ for hypotheses, 
$c_k\pi(W(r_k,s_k)) \exp\{i\vphi_k\}\in \gR(\partial_1 b)$ for every $k$, and thus we also have
$ z^{(\chi)}_{(f,g)}(b) V_{\partial_1 b} z^{(\chi)*}_{(f,g)}(b) \in  \gR(\partial_1 b)$. This entails
$$V_{\partial_0 b}  = z^{(\chi)}_{(f,g)}(b) V_{\partial_1 b} z^{(\chi)*}_{(f,g)}(b) \in \gR(\partial_1 b)\:.$$
Since $\partial_0b, \partial_1 b \in \Sigma_0$ and $V_{\partial_0 b}$ are generic, we have found that 
$V_a \in \cap_{I \in \cR} \gR(I) = \bC\ide$ in view of the irreducibility property of the pre-cosheaf.\\
$(b)$ Let us establish the equivalence of cocycles associated to different maps $\chi$ but with the same $(f,g)$.
 If $\chi_1$ and $\chi_2$ are defined as in (\ref{assignchi}), for every $I\in \cR$ the map
 $\Delta\chi^{a} \doteq \chi_1^{a}-\chi_2^{a}$ is smooth and compactly supported in the open set  $I$, so that it can  be extended uniquely
 as a smooth function over $\bS^1$ compactly supported in $I$. As usual,  we indicate by $\Delta \chi^{a}$ this unique extension.
Define the field of unitaries $V: \cR \ni I \mapsto V_I \doteq  Z\left( f\Delta \chi^{a} , g\Delta \chi^{a}\right) \in \gR(I)$.
For every $b\in \Sigma_1$, we get
\begin{align*}
V_{\partial_0b} z^{(\chi_1)}_{(f,g)}(b) &=  Z\left( f\Delta \chi , g\Delta \chi \right)
 Z\left((\chi_1^{(\partial_1 b)}- \chi_1^{(\partial_0 b)})f, (\chi_1^{(\partial_1 b)}- \chi_1^{(\partial_0 b)}) g\right)\\
 &=Z\left((\chi_1^{(\partial_1 b)}- \chi_1^{(\partial_0 b)} + \chi_1^{(\partial_0 b)} - \chi_2^{(\partial_0 b)})f, 
(\chi_1^{(\partial_1 b)}- \chi_1^{(\partial_0 b)} + \chi_1^{(\partial_0 b)} - \chi_2^{(\partial_0 b)}) g\right)\\
&= Z\left((\chi_1^{(\partial_1 b)}- \chi_2^{(\partial_0 b)})f, 
(\chi_1^{(\partial_1 b)} - \chi_2^{(\partial_0 b)}) g\right)
\end{align*}
where, passing from the first to the second line, we have omitted a phase arising from Weyl relations, since it vanishes as before.
With an analogous computation we similarly find:
$$ z^{(\chi_2)}_{(f,g)}(b) V_{\partial_1b} = Z\left((\chi_1^{(\partial_1 b)}- \chi_2^{(\partial_0 b)})f, 
(\chi_1^{(\partial_1 b)} - \chi_2^{(\partial_0 b)}) g\right)\:,$$
so that $V_{\partial_0b} z^{(\chi_1)}_{(f,g)}(b) =  z^{(\chi_2)}_{(f,g)}(b) V_{\partial_1b}$. \\
Let us pass to prove $(c)$. Within the hypotheses as in $(c)$ consider a positively oriented $1$-simplex
$b$ with $\overline{\partial_0 b}$ disjoint from  $\overline{\partial_1 b}$; let us indicate by $I_b\in \cR$ the open proper 
segment lying between $\partial_0 b$ and  $\partial_1 b$.
If  $z^{(\chi)}_{(f,g)}$ and $z^{(\chi)}_{(f',g')}$ are equivalent, we may write
$V_{\partial_0b}  =  z^{(\chi)}_{(f',g')}(b) V_{\partial_1b} z^{(\chi)*}_{(f,g)}(b)\:,$
 for some unitaries $V_{\partial_j b} \in \gR(\partial_jb)$. Therefore
 $V_{\partial_0b}  =  z^{(\chi)}_{(f',g')}(b)  z^{(\chi)*}_{(f,g)}(b) \widetilde{V}_{\partial_1b}\:,$ 
 and thus $V_{\partial_0b} \widetilde{V}_{\partial_1b}^*  =  z^{(\chi)}_{(f',g')}(b) z^{(\chi)*}_{(f,g)}(b)$
where we have introduced the unitary operator $\widetilde{V}_{\partial_1b}\doteq z^{(\chi)}_{(f,g)}(b)
 V_{\partial_1b} z^{(\chi)*}_{(f,g)}(b)$. As ${V_{\partial_1b}} z^{(\chi)}_{(f,g)}(b) \in \gR(\partial_1b)$, 
 following the argument as in the proof 
 of $(a)$, we achieve $\widetilde{V}_{\partial_1b} z^{(\chi)}_{(f,g)}(b) \in  \gR(\partial_1b)$ and so
 $\widetilde{V}_{\partial_1b} z^{(\chi)*}_{(f,g)}(b) \in  \gR(\partial_1b)$. 
 The term $z^{(\chi)}_{(f',g')}(b) z^{(\chi)*}_{(f,g)}(b)$ can be computed and, in view of Weyl relations, it finally arises
\beq V_{\partial_0b} \widetilde{V}_{\partial_1b}^*  =  z^{(\chi)}_{(f'-f,g'-g)}(b) e^{i\vphi} \:,\label{centralN}\eeq
where $\vphi \in \bR$ depends on $f,f',g,g', \chi$.
Now consider two real smooth functions $r,s$
supported in $I_b$. The Weyl generator $Z(r,s)$ belongs to $\gR(I_b)$ and thus it commutes with
both $\widetilde{V}_{\partial_1b}^*$ and $\widetilde{V}_{\partial_0b}$ so that (\ref{centralN}) produces (notice that $\chi^{(b)}=1$
on $I_b$)
\begin{align*}
V_{\partial_0b} \widetilde{V}_{\partial_1b}^*  &= Z(r,s) z^{(\chi)}_{(f'-f,g'-g)}  Z(r,s)^* e^{i\vphi}\\
&= z^{(\chi)}_{(f'-f,g'-g)} e^{i\vphi} \exp\left\{i \int_{\bS_1} ((f'-f)s - (g'-g)r) d\theta \right\}\:.
\end{align*}
Comparing with (\ref{centralN}) we conclude that it must be 
$$\exp\left\{i \int_{\bS_1} ((f'-f)s - (g'-g)r) d\theta \right\}  = 1\:.$$
Arbitrariness of the smooth functions $r,s$ implies that $f'-f= 0$, $g'- g =0$ on $I_b$. Since the procedure can be implemented choosing 
$I_b$ as a sufficiently small neighborhood of every point on $\bS^1$, we conclude that $f=f'$ and $g=g'$ everywhere on $\bS^1$.\\
Let us conclude the proof by demostrating statement $(d)$. Referring to Theorem \ref{thoremvacuum1} one finds 
\begin{align*}
U_{(r,0)} z^{\chi}_{(f,g)}(b) U_{(r,0)}^* &= e^{-ir P^\otimes} z^{\chi}_{(f,g)}(b) e^{ir P^\otimes}  \\
&=W\left[e^{-ir P^\otimes} 2^{-1/2} (A^{1/4} \chi^{(b)}  f + iA^{-1/4} \chi^{(b)} g)\right]\\
 &= W\left[ 2^{-1/2} A^{1/4} e^{-ir P}   \chi^{(b)}  f
+ i 2^{-1/2} A^{-1/4} e^{-ir P}  \chi^{(b)}  g\right]\:,
\end{align*}
where we used the fact that $e^{-ir P}$ and the spectral measure of $A$ commute.  On the other hand, one sees that $(e^{-ir P} h)(\theta) = h(\theta -r)\doteq (\beta^*_r(h))(\theta)$
for every $h\in L^2(\bS^1,d\theta)$, by working in Fourier representation. This, together with \eqref{covset}, implies the validity 
of the thesis imediately.
\end{proof}

 \subsection{Representations of $\pi_1(\bS^1)$} 
 Let us state and prove  some properties of the   representations of $\pi_1(\bS^1)$ associated to the previously constructed cocycles.
 
 \begin{theorem} \label{propcocycles} Consider the $1$-cocycle $z^{(\chi)}_{(f,g)}$ (\ref{cocyclef}) defined in Theorem \ref{propcocycle} and the associated 
 representation (\ref{piz}) of $\pi_1(\bS^1)\equiv \bZ$. The representation reads, in this case,
 \beq
 \pi_{(f,g)} : \bZ \ni n \mapsto Z(nf,ng) \label{repzf}\:,
 \eeq
 where it is manifest that it does not depend on the choice of $\chi$.
The following further facts hold.
\begin{itemize}
\item[$(a)$] 
$\pi_{(f,g)}$ is trivial
-- equivalently $z^{(\chi)}_{(f,g)}$ is trivial -- if and only if $f= g = 0$.\\
\item[$(b)$]  for every pair  $(f,g), (f',g') \in \cS\times \cS$ with $(f,g) \neq (0,0) \neq (f',g')$,
 the unitary representations $\pi_{(f,g)}$ and $\pi_{(f',g')}$ are unitarily
equivalent.\\
\item[$(c)$] For every $(f,g) \in \cS$,  every $r\in \bR$ and $b\in \Sigma_1(\cR)$, and very $n\in \bZ$, it holds
$$U_{(r,0)} \pi_{(f,g)}(n) U_{(r,0)}^* = \pi_{(\beta^*_r (f) ,\beta^*_r(g))} (n)\:,$$
where $U_{(r,0)}$ is the one-parameter unitary group implementing $\theta$-displacements $\beta_r$ and leaving the vacuum invariant, introduced in 
Theorem \ref{thoremvacuum1} and $\beta_r^*$ is the pull-back action of $\theta$-displacements on functions defined on $\bS^1$. 
\item[$(d)$] If $(0,0) \neq (f,g)\in \cS$, the space $\gH_0 = \gF(\cH)$ decomposes as a countably infinite Hilbert sum of closed pairwise orthogonal 
subspaces
$
\gH_0 = \bigoplus_{k=0}^{+\infty} {\gH_0}^{(f,g)}_k
$
such that the following holds for $k\in \bN$.
\begin{itemize}
\item[$(i)$] ${\gH_0}^{(f,g)}_k$ is invariant under $\pi_{(f,g)}$.

\item[$(ii)$] There is a unitary map $U^{(f)}_k : {\gH_0}^{(f)}_k \to L^2(\bR,dx)$ such that $\pi_{(f,g)} \rest_{{\gH_0}_k}$
admits a direct integral  decomposition into one-dimensional irreducible representations $\lambda_x$ of $\bZ$ (\ref{lambdax})
as
\beq U^{(f)*}_k \pi_{(f,g)} \rest_{{\gH_0}_k} U^{(f,g)}_k = \int_{\bR}^{\oplus} dx\:\: \lambda_x  \:, \label{dirdec}\eeq
where $L^2(\bR,dx)= \int_{\bR}^{\oplus} dx\:\:\cH_x $, with $\cH_x \doteq \bC$ and $dx$ being the Lebesgue measure on $\bR$.
\end{itemize}
\end{itemize}
\end{theorem}

\begin{proof}
Let us first prove \eqref{repzf} for $n=1$. Since we know that \eqref{piz} gives rise to a group representation of $\pi_1(\bS)$
when $z^{(\chi)}_{(f,g)}$ is a cocycle, to  prove \eqref{repzf} for  $n=1 \in \bZ = \pi_1(\bS^1)$
 i.e.,  
 \beq z^{(\chi)}_{(f,g)}(p) = Z(f,g)\:, \quad \mbox{for $p\in 1$} \label{z1}\eeq
 it is enough to prove it for a fixed path  $p\in 1$, because the result must not depend on the particular 
 path in $1$. 
 To this end, if $\bS^1 = [-\pi,\pi]$ where $-\pi\equiv \pi$, consider the path $p\in 1$
made of the 1-simplices $b$, with
$|b|\doteq (-\frac{\pi}{2}-\epsilon, \frac{\pi}{2} + \epsilon)$, $\partial_1 b \doteq (-\frac{\pi}{2}-\epsilon, -\frac{\pi}{2} + \epsilon)$, 
$\partial_0 b \doteq (\frac{\pi}{2}-\epsilon, \frac{\pi}{2} + \epsilon)$
and  $b'$ with $|b'|\doteq (\frac{\pi}{2}-\epsilon, \pi] \cup [-\pi, -\frac{\pi}{2} + \epsilon)$,
$\partial_1 b' \doteq (\frac{\pi}{2}-\epsilon, \frac{\pi}{2} + \epsilon)$, $\partial_0 b' \doteq (-\frac{\pi}{2}-\epsilon, -\frac{\pi}{2} + \epsilon)$, where $\epsilon>0$ is so small that $\partial_0 b \cap \partial_1 b =\emptyset$. Using the definition of $\chi^{(b)}$
and $\chi^{(b')}$, it follows immediately that $\chi^{(b)} + \chi^{(b')} =1$ everywhere on $\bS^1$. Therefore we have that $z^(\chi)_{(f,g)}(p)$ equals 
\begin{align*}
Z(\chi^{(b')}f,\chi^{(b')}g)Z(\chi^{(b)}f,\chi^{(b)}g) &= 
Z((\chi^{(b')}+ \chi^{(b)})f,(\chi^{(b')}+ \chi^{(b)})g) e^{i\int_{\bS^1}
\chi^{(b')}\chi^{(b)}(fg-gf) d\theta}\\
&= Z(f,g)\:.
\end{align*}
We have established \eqref{z1}, i.e.~\eqref{piz} for $n=1$. Let us generalize the result for $n\in \bZ$.
By the definition of $Z$ and making  use of Weyl commutation 
relations one gets  \beq
Z(nf,ng)Z(mf,mg) = Z((n+m)f, (n+m)g)\:, \forall n,m\in \bZ\label{abelrel}\:.\eeq
Using the fact that $\pi^(\chi)_{z_{(f,g)}}$ as defined in \eqref{piz} is a group representation
of  $\pi_1(\bS^1)=\bZ$, which is Abelian and generated by $1$, one has that \eqref{abelrel} and
\eqref{z1} together  yield \eqref{piz} in the general case.\\
 Let us pass to prove $(a)$.
 As a consequence of \eqref{repzf},
  it is clear that this representation is trivial, that is $z^{(\chi)}_{(f,g)}$ is such, due to Theorem \ref{teoruzzi},
 if and only if $Z(f,g)=\ide$. It is equivalent to say $W\left[2^{-1/2}(A^{1/4}f+ iA^{-1/4} g)\right] = \ide$. By Theorem \ref{thoremvacuum1}
 we know that
 $$\left\langle \Psi, W\left[2^{-1/2} A^{1/4} f +  i2^{-1/2} A^{-1/4} g\right] \Psi \right\rangle = e^{-\frac{1}{4}\left(\langle f, A^{1/2} f\rangle
 + \langle g, A^{-1/2} g\rangle
  \right)} \quad \mbox{for all $f,g \in C^\infty(\bS^1, \bR)$}\:.$$
 Since $||\Psi||=1$ we have finally that $Z(f,g)=\ide$ entails  $\langle f, A^{1/2} f\rangle
 + \langle g, A^{-1/2} g\rangle =0$ and so $f,g =0$ because $A^{-1/4}$ and $A^{1/4}$
 are strictly positive. We have found that triviality of $z^\chi_{(f,g)}$ implies $f,g=0$. The converse is obvious and so
  the proof of $(a)$ is concluded. \\
  Let us demonstrate $(b)$.
 Assume $||2^{-1/2}(A^{1/4}f + iA^{-1/4}g)|| = ||2^{-1/2}(A^{1/4}f' + iA^{-1/4}g')|| = a \neq 0$ (the case equal to $0$ being obvious). Defining
 $\psi_1 \doteq 2^{-1/2}(A^{1/4}f + iA^{-1/4}g)$
 we can complete this vector to a maximal orthogonal system $\{\psi_n\}_{n\in \bN}$ of $L^2(\bS^1,d\theta)$
where $||\psi_n||=a$ for every $n\in \bN$. Similarly, defining $\phi_1 \doteq 2^{-1/2}(A^{1/4}f' + iA^{-1/4}g')$
 we can complete this vector to a maximal orthogonal system $\{\phi_n\}_{n\in \bN}$ of $L^2(\bS^1,d\theta)$,
where $||\phi_n||=a$ for every $n\in \bN$.  There is a unique unitary operator $U: L^2(\bS^1,d\theta) \to L^2(\bS^1,d\theta)$
completely individuated by the requirements $U\psi_n = \phi_n$ for every $n\in \bN$.
It is a known property of Weyl generators $W[\psi] = e^{\overline{a(\psi)-a^*(\psi)}}$ that
$$V_\otimes W[\psi] V_\otimes^* = W[V\psi]$$
where the unitary operator $V_\otimes$ in the Fock space is defined by tensorialization of the unitary operator $V$
in the one-particle
space, with the requirement that $V_\otimes$ reduces to the identity acting on the vacuum vector.
As a consequence $U_\otimes W[\psi_1] U_\otimes^* = W[U\psi_1] = W[\phi_1]$ or, equivalently, $U_\otimes Z(f,g) U^*_\otimes
= Z(f',g')$ and thus $U_\otimes Z(nf,ng) U^*_\otimes
= Z(nf',ng')$, making use of (\ref{abelrel}). We have found that $\pi_{(f,g)}$ and $\pi_{(f',g')}$ are unitarily equivalent.
 Let us pass to the case
 $0 \neq ||2^{-1/2}(A^{1/4}f + iA^{-1/4}g)|| \neq ||2^{-1/2}(A^{1/4}f' + iA^{-1/4}g')||  \neq 0$
 and define the real number
$r \doteq ||2^{-1/2}(A^{1/4}f + iA^{-1/4}g)||/||2^{-1/2}(A^{1/4}f' + iA^{-1/4}g')||$.
With the procedure used in the former case one achieves the existence of a unitary operator $V$ on the Fock space such that
$$W\left[2^{-1/2}(A^{1/4}f + iA^{-1/4}g)\right] = VW\left[r2^{-1/2}(A^{1/4}f' + iA^{-1/4}g')\right]V^*\:.$$
To conclude it is sufficient to establish the existence of a second unitary operator $E$ (depending on the considered $g$ and $r$)
with
$$W\left[2^{-1/2}(A^{1/4}f' + iA^{-1/4}g')\right] = EW\left[r \: 2^{-1/2}(A^{1/4}f' + iA^{-1/4}g')\right]E^*\:.$$
This fact is an immediate consequence of the following result proved in the Appendix A. \\

\begin{lemma}\label{lemmaE} Let $\cH$ be a complex Hilbert space with associated bosonic Fock space $\gF_+(\cH)$. Define the unitary Weyl
generators $W[\psi]$ as in (\ref{weylgen}) for every $\psi \in \cH$. For every fixed $\psi\in \cH$ with $||\psi||=1$,
there is a strongly
continuous one-parameter group of unitary operators $\{E^{(\psi)}_\lambda\}_{\lambda\in \bR}$ such that
\beq
E^{(\psi)}_\lambda W[\psi] E^{(\psi)*}_\lambda = W\left[e^\lambda\psi\right]\:, \quad \mbox{for all $\lambda \in \bR$.}
\label{exp}
\eeq
\end{lemma}
The proof of $(c)$ follows immediately from $(d)$ in Theorem \ref{propcocycle} taking the independence from $\chi$
into account.\\
Finally we prove $(d)$. Fix $(f,g)\in \cS$. In view of the Weyl commutation relations for operators
$W[\psi]$,  the unitary operators
$$U(a,b)\doteq Z\left(\frac{(a+ib)f}{|| A^{1/4} f + iA^{-1/4}g||}, \frac{(a+ib)g}{|| A^{1/4} f + iA^{-1/4}g||}\right)\:,\quad  (a,b)\in \bR^2\:,$$ fulfill the one-dimensional Weyl  relations
$$U(a,b)U(a',b')= U(a+a',b+b') e^{-i(ab'-a'b)/2}\:,\quad U(a,b)^* = U(-a,-b)\:.$$
Due to the uniqueness property in the Stone - von Neumann -Mackey Theorem, the space $\gH_0$ decompose into a direct sum of 
pairwise orthogonal closed subspaces ${\gH_0}_k$ where each ${\gH_0}_k$ is unitarily equivalent to $L^2(\bR,dx)$ and the relevant 
unitary map satisfies $$U^{(f,g)*}_k U(a,b) \rest_{{\gH_0}_k} U^{(f,g)}_k = \exp{i\{\overline{aX+bP}\}}\:,$$ $X,P$ being the standard position and momentum
operators on the real line ($aX+bP$ is defined on the core given by the Schwartz space). As a consequence
$$U^{(f,g)*}_k \pi_{(f,g)}(n) \rest_{{\gH_0}_k} U^{(f,g)}_k = U^{(f,g)*}_k Z(nf,ng) U^{(f,g)}_k = e^{in c X}\:, $$
with $c= ||2^{-1/2} (A^{1/4} f + i A^{-1/4} g)||>0$ constant.
Then the spectral decomposition of $cX$ gives rise to \eqref{dirdec} immediately. To end the proof the only thing to show  is that the number of spaces ${\gH_0}_k$ is infinite. Since $\gH_0$ is separable that infinite must be countable.
It is known by the general theory of Weyl algebras on finite-dimensional symplectic spaces that the spaces ${\gH_0}_k$ can be obtained as follows. Using weak operator topology, define the operator
$$P=\frac{1}{2\pi}\int_{\bR^2} e^{- (u^2 +v^2)/4} U(u,v)\: dudv$$ which turns out to be a nonvanishing orthogonal projector.
If $\{\phi_k\}_{k\in G}$ is a Hilbert basis for the subspace $P(\gH_0)$, for any fixed $k\in G$, ${\gH_0}_k$ is the closed space generated by all of $U(a,b)\phi_k$ as $a,b \in \bR$. To conclude it is sufficient to prove that $G$ must be infinite. To this end consider a Hilbert basis in $\cH$, 
$\psi_1= (A^{1/4} f + i A^{-1/4} g)/||A^{1/4} f + i A^{-1/4} g||$, $\psi_2$, $\psi_3$, $\ldots$ and an associated orthonormal (not necessarily complete) system in $\gH_0$:
$\Psi_1\doteq\Psi$ (the vacuum), $\Psi_2 \doteq a^*(\psi_2)\Psi$, $\Psi_3 \doteq a^*(\psi_3)\Psi$, $\ldots$.
 By construction, one can verify that
 $$(P \Psi_h|P\Psi_k) = \frac{1}{2\pi}\int_{\bR^2} e^{- (u^2 +v^2)/4} (\Psi_h|U(u,v)\Psi_k)\: dudv
= \frac{\delta_{hk}}{2\pi}\int_{\bR^2} e^{- (u^2 +v^2)/4} e^{- (u^2 +v^2)/4}\: dudv\:.$$
Therefore, up to normalization, $P \Psi_1, P\Psi_2, \ldots \in P(\gH_0)$ is an infinite orthonormal system in $P(\gH_0)$. This means that
$P(\gH_0)$ admits an infinite Hilbert base.
\end{proof}

 \subsection{Examples of topological superselection sectors}
 In this section we show how to associate every localized cocycles $z^{(\chi)}_{(f,g)}$
  with a sharp excitation of the reference vacuum representation $\{{\bpi_0}, \ide\}$. In case we have a pair of unitarily inequivalent  cocycles, they would provide with a pair of  unitarily 
 inequivalent generalized representations fulfilling the selection criterion, and thus two different superselection sectors. The idea is similar to that exploited to define a relevant functor in the proof of Theorem 4.3 in \cite{BR2}. However there are two important differences. First of all, here we are dealing with a proper subset of 
cocycles and not with the whole category $Z^1(\gR)$. Secondly, as we shall see into details shortly, the map that associates 
cocycles to generalized representations in the proof of Theorem 4.3 in \cite{BR2} does not work in our lower dimensional case and needs a modification. \\
Consider $z^{(\chi)}_{(f,g)} \in Z^1(\gR)$ and, for $I, \widetilde{I} \in \cR$ with $\widetilde{I} \subseteq I$ define
\beq
\pi^{z^{(\chi)}_{(f,g)}}_I(A) &\doteq& z^{(\chi)}_{(f,g)}(b_I) {{\bpi_0}}_I(A) z^{(\chi)}_{(f,g)}(b_I)^*\:, \quad A \in \cW(I)\label{1pi}\\
\psi^{z^{(\chi)}_{(f,g)}}_{I,\tilde{I}} &\doteq& z^{(\chi)}_{(f,g)}\left(I,\widetilde{I}\right) \label{1psi}\:.
\eeq
above, $b_I$ is a $1$-simplex with final point $\partial_0 b_I \doteq I$ and initial point $\partial_1 b_I \doteq J$ where $J\subset I'$ and, finally,  $b_I$ is positively oriented  w.r.t. the chose orientation of $\bS^1$; the $1$-simplex  $(I,\widetilde{I})$
is that with $\partial_1\left(I,\widetilde{I}\right) = \widetilde{I}$ and  $\partial_0\left(I,\widetilde{I}\right) =I =
\left|\left(I,\widetilde{I}\right)\right|$.\\
Finally define
\beq \pi^{z^{(\chi)}_{(f,g)}} : \cR \ni I \mapsto \pi^{z^{(\chi)}_{(f,g)}}_I \:, \quad\mbox{and}
\quad \psi^{z^{(\chi)}_{(f,g)}} : \cR\times \cR \ni (I,\widetilde{I}) \mapsto \psi^{z^{(\chi)}_{(f,g)}}_{I,\tilde{I}}\quad 
\mbox{for $I,\widetilde{I} \in \cR$ and 
$\widetilde{I} \subseteq I$} \label{1pipsi}\:.\eeq
We are going to establish that $\{\pi^{z^{(\chi)}_{(f,g)}}, \psi^{z^{(\chi)}_{(f,g)}}\}$ is a net representation which satisfies 
the selection criterion. \\

\begin{remark} The definition given in (\ref{1pi}) and (\ref{1psi}) are the same as that used in Theorem 4.3 
in \cite{BR2} with the only difference that $b_I$ is now a (positively oriented) $1$-simplex  rather than a path. This is due to the fact that, if we adopted the definition as in \cite{BR2}, the defined objects would depend on the chosen path, differently from the higher dimensional case. We shall come back to this issue later.
\end{remark}

\noindent We have the following theorem which explain how to associate cocycles $z^{(\chi)}_{(f,g)}$ with net representations 
verifying the selection criterion introduced above.\\

\begin{theorem}If $z^{(\chi)}_{(f,g)}\in Z^1(\gR)$,  the  pair
 $\{\pi^{z^{(\chi)}_{(f,g)}}, \psi^{z^{(\chi)}_{(f,g)}}\}$   defined as in (\ref{1pipsi})
  is a unitary
  net representation
 of $\cW$ over $\gH_0$, which is independent from the choice of the simplices $b_I$ adopted in (\ref{1pi}).
 The further following results hold true.
 \begin{itemize}
\item[$(a)$] $\{\pi^{z^{(\chi)}_{(f,g)}}, \psi^{z^{(\chi)}_{(f,g)}}\}$ is irreducible and satisfies the selection criterion and thus defines a 
sharp excitation of the reference vacuum net representation $\{{\bpi_0}, \bI\}$, giving rise to a superselection sector.
\item[$(b)$] $\{\pi^{z^{(\chi)}_{(f,g)}}, \psi^{z^{(\chi)}_{(f,g)}}\}$ and 
$\{\pi^{z^{(\chi')}_{(f',g')}}, \psi^{z^{(\chi')}_{(f',g')}}\}$ belong to the same superselection sector 
(i.e. they are unitarily equivalent) if and only if  $f=f'$ and $g=g'$. 
\item[$(c)$] The $1$--cocycle associated with the net representation $\{\pi^{z^{(\chi)}_{(f,g)}}, \psi^{z^{(\chi)}_{(f,g)}}\}$ 
as in (\ref{zetapi}) coincides with $z^{(\chi)}_{(f,g)}$ itself.
\end{itemize}
\end{theorem}

\begin{proof}
 First of all we have to show that (\ref{UNR}) are fulfilled. By direct inspection, exploiting the definition
of $z^{(\chi)}_{(f,g)}$, we find
\beq
 \psi^{z^{(\chi)}_{(f,g)}}_{I,\tilde{I}} =
  Z\left((\chi^{(\tilde{I})}- \chi^{(I)})f,(\chi^{(\tilde{I})}- \chi^{(I)})g \right) \label{zz}\:,
\eeq
where the function $\chi^{(\tilde{I})}$ has  been extended to the whole larger interval $I$
 as a constant function as beforehand, and similarly, the so 
 obtained function $\chi^{(\tilde{I})}- \chi^{(I)}$, which is compactly supported in $I$, has been extended to the null 
 function outside $I$. With this definition the second identity in (\ref{UNR}) arises from (\ref{zz}) and Weyl identities 
 straightforwardly. Let us pass to the former identity in (\ref{UNR}).
By linearity and continuity,  
this can be done by verifying the first statement in (\ref{UNR})
with the involved function applied to local 
Weyl generators $A= W(\Phi,\Pi)$ with $\Phi,\Pi$ supported in ${I}$. Remind that, in our case, $j_{I\tilde{I}}$
can be omitted interpreting the elements of the local Weyl algebras working as elements of the global Weyl algebra $\cW$.
 By direct inspection, employing ${\bpi_0}\left(W(\Phi,\Pi)\right) = Z\left(\Phi,\Pi\right)$, 
  making use of Weyl relations and employing the definition of $z^{(\chi)}_{(f,g)}$
 one finds that, if $\Phi,\Pi$ are supported in ${I}$,
 \beq \pi^{z^{(\chi)}_{(f,g)}}_{I}\left( W(\Phi,\Pi)\right)= Z(\Phi,\Pi)  
 \exp\left\{i\sigma\left((\Phi,\Pi),((1-\chi^{(I)}) f,(1-\chi^{(I)}) g)\right)\right\}\:.\label{pp}\eeq
 Notice that only $I$ appears in the right-hand side, so that different choices for $b_{I}$ yields 
the same result and the choice of $b_{I}$ is immaterial.
\eqref{pp} and \eqref{zz} entail, in view of Weyl identities
\begin{align*}
\psi^{z^{(\chi)}_{(f,g)}}_{I,\tilde{I}} \pi^{z^{(\chi)}_{(f,g)}}_{\tilde{I}}&\left( W(\Phi,\Pi)\right)
\psi^{z^{(\chi)}_{(f,g)}*}_{I,\tilde{I}} \\
&= 
Z(\Phi,\Pi)  
 e^{i\sigma\left((\Phi,\Pi),((1-\chi^{(\tilde{I})}) f,(1-\chi^{(\tilde{I})}) g)\right)}
  e^{i\sigma\left((\Phi,\Pi),((\chi^{(\tilde{I})}-\chi^{(I)}) f,(\chi^{(\tilde{I})}-\chi^{(I)}) g)\right)}\\
&= Z(\Phi,\Pi)  
 e^{i\sigma\left((\Phi,\Pi),((1-\chi^{(I)}) f,(1-\chi^{(I)}) g)\right)}\\
&= \pi^{z^{(\chi)}_{(f,g)}}_{I}\left( W(\Phi,\Pi)\right)\:.
\end{align*}
This result implies the first identity in \eqref{UNR}.\\
Let us prove 
$(a)$. If $O \in \cR$ let  $N \subset \bS^1$ a (connected) simply connected open set (so that either $N\in \cR$
or $N= \bS^1 \setminus \{p\}$ for some $p\in \bS^1$) with $\overline{O} \subset N$.
Fix $I\in \cR$ with both $\overline{I} \subset N$ and $I \subset O'$. We can define $W_I^{NO}$ as
\beq
W_I^{NO} \doteq z^{(\chi)}_{(f,g)}(b_I)^*\label{WINO}\:,
\eeq
where $b_I\in \Sigma_1(\cR)$ is chosen as in the \eqref{1pi}
but  $|b_I| \subset N$. With the definition \eqref{WINO} the three requirements under \eqref{SS} turn out to be valid.
The first requirement is verified automatically in view of \eqref{1pi}, the remaining two have straightforward proofs
based on  Weyl relations and proceeding as above. 
The proof of the irreducibility of $\{\pi^{z^{(\chi)}_{(f,g)}}, \psi^{z^{(\chi)}_{(f,g)}}\}$ will be postponed at the end of the proof of $(b)$.\\
$(b)$ In view of $(b)$ and $(c)$ in Theorem \ref{propcocycle}, the thesis  is equivalent to say that 
$\{\pi^{z^{(\chi)}_{(f,g)}}, \psi^{z^{(\chi)}_{(f,g)}}\}$ and $\{\pi^{z^{(\chi')}_{(f',g')}}, \psi^{z^{(\chi')}_{(f',g')}}\}$
are unitary equivalent if and only if $z^{(\chi)}_{(f,g)}$ and $z^{(\chi')}_{(f',g')}$ are unitarily equivalent. Let us prove the thesis in this 
second form. Suppose that $T \in (z^{(\chi)}_{(f,g)}, z^{(\chi')}_{(f',g')})$ is unitary, as a consequence
$T \in (\{\pi^{z^{(\chi)}_{(f,g)}}, \psi^{z^{(\chi)}_{(f,g)}}\},\{\pi^{z^{(\chi')}_{(f',g')}}, \psi^{z^{(\chi')}_{(f',g')}}\})$.
Indeed take $A\in\cW(I)$ and remind that $T_{\partial_1b_I}\in \gR(\partial_1b_I)$ and thus 
$T_{\partial_1b_I}$ and $T^*_{\partial_1b_I}$
 commute
with $\bpi_0(A)$ because $\partial_1b_I \subseteq I'$. Hence,

\begin{align*}
T_I \ \pi^{z^{(\chi)}_{(f,g)}}(A)\ T^*_I &= T_I\ z^{(\chi)}_{(f,g)}(b_I) {\bpi_0}(A)
z^{(\chi)*}_{(f,g)}(b_I)\ T^*_I \\
&=  z^{(\chi ')}_{(f',g')}(b_I) T_{\partial_1b_I}{\bpi_0}(A)
(T_I z^{(\chi)}_{(f,g)}(b_I))^*\\
&= z^{(\chi ')}_{(f',g')}(b_I) {\bpi_0}(A)
 T_{\partial_1b_I} (z^{(\chi ')}_{(f',g')}(b_I) T_{\partial_1b_I})^*\\
&= z^{(\chi ')}_{(f',g')}(b_I) {\bpi_0}(A)
 T_{\partial_1b_I} T^*_{\partial_1b_I} z^{(\chi ')*}_{(f',g')}(b_I)\\
&= z^{(\chi ')}_{(f',g')}(b_I) {\bpi_0}(A) z^{(\chi ')*}_{(f',g')}(b_I) \\
&= \pi^{z^{(\chi')}_{(f',g')}}(A)\:.
\end{align*}
Similarly, directly by the definition of $\psi^{z^{(\chi)}_{(f,g)}}$ one also gets, if $\widetilde{I}\subseteq I$,
$$T_{I}\ \psi^{z^{(\chi)}_{(f,g)}}_{I,\tilde{I}} =
T_I\ z^{(\chi)}_{(f,g)}(I,\widetilde{I}) = z^{(\chi')}_{(f',g')}(I,\widetilde{I})\  T_{\tilde{I}}
=  \psi^{z^{(\chi')}_{(f',g')}}_{I,\tilde{I}}\ T_{\tilde{I}}\:.$$
The obtained result implies that equivalence of cocycles entails unitary
 equivalence of the associated generalized representations. 
Let us prove the converse. To this end suppose that 
$$T \in (\{\pi^{z^{(\chi)}_{(f,g)}}, \psi^{z^{(\chi)}_{(f,g)}}\},\{\pi^{z^{(\chi')}_{(f',g')}}, \psi^{z^{(\chi')}_{(f',g')}}\})$$
is unitary.
For every $I\in \cR$ define the unitary operator
\beq
t_I \doteq z^{(\chi')}_{(f',g')}(b_{OI})^* \ T_O\  z^{(\chi)}_{(f,g)}(b_{OI})\:,
\eeq
where $b_{OI}\in \Sigma_1(\cR)$ is a positive oriented simplex 
such that $\partial_1  b_{OI} =I$, $\partial_0  b_{OI} =O$ and $I\subseteq O'$. We want to prove that $t_I$ defines a 
localized intertwiner for the cocycles associated with the representations we are considering.
First of all we notice that $t_I$ does not depend on the chosen $O\subset I'$ because, if $\widetilde{O} \subseteq O$ one has
\begin{align*}
t_I &\doteq z^{(\chi')}_{(f',g')}(b_{OI})^* T_O z^{(\chi)}_{(f,g)}(b_{OI})\\
& =
z^{(\chi')}_{(f',g')}(b_{OI})^* T_O z^{(\chi)}_{(f,g)}(b_{O\tilde{O}})
z^{(\chi)}_{(f,g)}(b_{\tilde{O}I})\\
&= z^{(\chi')}_{(f',g')}(b_{OI})^* T_O \psi^{z^{(\chi)}_{(f,g)}}(b_{O\tilde{O}})
z^{(\chi)}_{(f,g)}(b_{\tilde{O}I})\\
& = z^{(\chi')}_{(f',g')}(b_{OI})^*  \psi^{z^{(\chi')}_{(f',g')}}(b_{O\tilde{O}})
T_{\tilde{O}} z^{(\chi)}_{(f,g)}(b_{\tilde{O}I})\\
&= z^{(\chi')}_{(f',g')}(b_{OI})^* z^{(\chi')}_{(f',g')}(b_{O\tilde{O}})
T_{\tilde{O}} z^{(\chi)}_{(f,g)}(b_{\tilde{O}I})\\
 &= z^{(\chi')}_{(f',g')}(b_{\tilde{O}I})^* 
T_{\tilde{O}} z^{(\chi)}_{(f,g)}(b_{\tilde{O}I})\:.
\end{align*}
Using a suitable chain of $1$-simplices and using the identity above, one can pass from the initial $O \subseteq I'$ to any other $O_1\subseteq I'$. Now notice that, if $B\in \cW(O)$
\begin{align*}
t_{I} {\bpi_0}_O(B) &= z^{(\chi')}_{(f',g')}(b_{OI})^* T_O z^{(\chi)}_{(f,g)}(b_{OI}){{\bpi_0}}_O(B)  z^{(\chi)*}_{(f,g)}(b_{OI})\\
&= z^{(\chi')}_{(f',g')}(b_{OI})^* T_O  \pi^{z^{(\chi)}_{(f,g)}}_O(B) z^{(\chi)*}_{(f,g)}(b_{OI})\\
&=z^{(\chi')}_{(f',g')}(b_{OI})^*  \pi^{z^{(\chi')}_{(f',g')}}_O(B) T_O z^{(\chi)*}_{(f,g)}(b_{OI}) \\
& = {\bpi_0}_O(B) t_I\:.
\end{align*}
So $t_I(A) \in {\bpi_0}_O(\cW(O))'$.
By Haag duality, and using the fact that $O\subseteq I'$ is generic, we conclude that $t_I(A) \in {\bpi_0}_I(\cW(I))'' = \gR(I)$ for every
$A \in \cW(I)$,
as wanted. Finally, let us prove that $t$ is an intertwiner between cocycles. Consider $b\in \Sigma_1(\cR)$ with $\partial_0b = I$.
Fix $O\subseteq |b|'$ in such a way that there are two 
 positively oriented $1$-simplices with $O$ as end point and starting, respectively, from $\partial_0b$ and $\partial_1b$.
 Then we can write
\begin{align*}
t_{\partial_0b} z^{(\chi)}_{(f,g)}(b) &= z^{(\chi')}_{(f',g')}(b_{O\partial_0b})^* T_O z^{(\chi)}_{(f,g)}(b_{O\partial_0b}) z^{(\chi)}_{(f,g)}(b)\\
& = z^{(\chi')}_{(f',g')}(b_{O\partial_0b})^* T_O z^{(\chi)}_{(f,g)}(b_{O\partial_1b})\\
&=z^{(\chi')}_{(f',g')}(b_{O\partial_0b})^* z^{(\chi')}_{(f',g')}(b_{O\partial_1b}) z^{(\chi')*}_{(f',g')}(b_{O\partial_1b}) T_O z^{(\chi)}_{(f,g)}(b_{O\partial_1b})\\
&= z^{(\chi')}_{(f',g')}(b_{O\partial_0b})^* z^{(\chi')}_{(f',g')}(b_{O\partial_1b})  t_{\partial_1b}\\
&=z^{(\chi')}_{(f',g')}(b)t_{\partial_1b}\:.
\end{align*}
Let us now prove, as claimed, that $\{\pi^{z^{(\chi)}_{(f,g)}}, \psi^{z^{(\chi)}_{(f,g)}}\}$ is irreducible.
Suppose there is a unitary intertwiner $U \in (\{\pi^{z^{(\chi)}_{(f,g)}}, \psi^{z^{(\chi)}_{(f,g)}}\},\{\pi^{z^{(\chi)}_{(f,g)}}, \psi^{z^{(\chi)}_{(f,g)}}\})$\:.
As a consequence the  operators $t_I\doteq z^{(\chi)}_{(f,g)}(b_{OI})^* U_O z^{(\chi)}_{(f,g)}(b_{OI})$, where $I,O\in \cR$, 
$O \subseteq I'$ and the direction from $I$ to $O$ is positive, define a unitary intertwiner
$t\in (z^{(\chi)}_{(f,g)}, z^{(\chi)}_{(f,g)})$. The statement $(a)$ in Theorem \ref{propcocycle}
implies that the $t_I$ are all of the form $c \ide$ with $c\in \bC$ and $|c|=1$. Therefore the $U_O$ have the same form and
since $O$ can be chosen arbitrarily in $\cR$, 
$\{\pi^{z^{(\chi)}_{(f,g)}}, \psi^{z^{(\chi)}_{(f,g)}}\}$ is irreducible.\\
$(c)$ The statement is an immediate consequence of (\ref{zz}) and Weyl relations.
\end{proof}

\begin{remark} The definition of $\{ \pi^{z_{(f,g)}^{(\chi)}}, \psi^{z_{(f,g)}^{(\chi)}}\}$ can be modified 
changing the requirements on the simplex $b_I$. These changes do not affect the results in $4$ dimension as 
established in Theorem 4.3 in \cite{BR2} where $b_I$ can be replaced by any path $p_I$ ending on $I$
but starting from $\partial_1p \subseteq I'$. Remarkably,  the situation  is different here.
Replacing the 1-simplex $b_I$ in (\ref{1pi}) with a path $p_I$
ending in $I$ which winds $n\in \bZ$ times around the circle before reaching $I$ and such that the final $1$-simplex
ending on $I$ is positively oriented, with the initial point in $I'$, 
\beq
\rho^{z^{(\chi)}_{(f,g)}}_I(A) &\doteq& z^{(\chi)}_{(f,g)}(p_I) {\bpi_0}_I(A) z^{(\chi)}_{(f,g)}(p_I)^*\:, \quad A \in \cW(I)\label{1pi'}\\
\phi^{z^{(\chi)}_{(f,g)}}_{I,\tilde{I}} &\doteq& z^{(\chi)}_{(f,g)}\left(I,\widetilde{I}\right) \label{1psi'}\:.
\eeq
define a generalized representation which is not  in the class of representations considered in the theorem just proved.
 However this new representation turns out to be unitarily equivalent to  $\{{\pi}^{z_{(f,g)}^{(\chi)}}, {\psi}^{z_{(f,g)}^{(\chi)}}\}$, 
\beq
{\rho}^{z_{(f,g)}^{(\chi)}}(A) =  Z(nf,ng) \pi^{z_{(f,g)}^{(\chi)}}(A) Z(nf,ng)^*\:, \quad  
{\phi}^{z_{(f,g)}^{(\chi)}} = Z(nf,ng)^* {\psi}^{z_{(f,g)}^{(\chi)}} Z(nf,ng)= {\psi}^{z_{(f,g)}^{(\chi)}}\:.
\eeq
Another, more radical change may be performed in the definition \eqref{1pi}, if one assumes that the $1$-simplex $b_I$ with end points  $I$ and $J$ 
is {\em negatively} oriented. In this case one is committed to  replace also $z^{(\chi)}_{(f,g)}\left(I,\widetilde{I}\right)$ with 
$z^{(\chi)*}_{(f,g)}\left(I,\widetilde{I}\right)$ in the definition (\ref{1psi}), in order to 
obtain a generalized representation. With these changes definitions \eqref{1pi} and \eqref{1psi} work anyway 
and give rise to a different representation
$\{ \tilde{\pi}^{z_{(f,g)}^{(\chi)}}, \tilde{\psi}^{z_{(f,g)}^{(\chi)}}\}$.
Also this representation is not included in the class of representations considered in the theorem.
However that new representation is globally unitarily equivalent to a representation as those in the theorem, {\em but associated with a different cocycle}.
In fact it turns out to be unitarily equivalent to $\{\pi^{z_{(-f,-g)}^{(\chi)}}, \psi^{z_{(-f,-g)}^{(\chi)}}\}$, where we stress that the signs in front of $f$ and $g$, and thus the cocycle, has changed.
Indeed, one finds after a trivial computation based on the explicit form of cocycles:
\beq
\tilde{\pi}^{z_{(f,g)}^{(\chi)}}(A) =  Z(f,g) \pi^{z_{(-f,-g)}^{(\chi)}}(A) Z(f,g)^*\:, \quad  
\tilde{\psi}^{z_{(f,g)}^{(\chi)}} = Z(f,g)^* \tilde{\psi}^{z_{(-f,-g)}^{(\chi)}} Z(f,g)
= \tilde{\psi}^{z_{(-f,-g)}^{(\chi)}}\:.
\eeq
\end{remark}

\section{Conclusions and outlook}
In this paper we showed the first direct construction of 1-cocycles of topological nature, originally defined abstractly in four dimensional spacetimes in \cite{BR2}, in what we believe to be the easiest possible case, namely massive free scalar free fields on two dimensional Einstein spacetime. 

Contrary to the theorem proved in \cite{BR2} for the case of abelian fundamental groups, in our situation the constructed 1-cocycles are not just characters of the group. We addressed ourselves to the very preliminary and basic constructions, and we left open many questions, like the completeness of the found sectors, the relation between the category of sectors and that of generalized representations satisfying the selection requirement \eqref{SS}, whose proof of equivalence \cite{BR2} holds only in the four dimensional case, and many other possibilities.
   
\noindent This opens the door to many new directions of research. The one we are trying first is on the investigation of the full spacetime construction. It requires some variations from what we discussed in the body of the paper. The second possible direction is on trying to see whether we can reach the completeness of the topological superselection sectors. A third one consists in generalizing the construction to higher dimensions, in both the abelian and non-abelian cases of fundamental groups of the Cauchy surfaces. Another would be the investigation of the case of charged bosons.  A more ambitious goal would be to export the construction of topological cocycles in the case of massive interacting quantum field theories on Einstein or the two dimensional de Sitter spacetimes. 
 
\section*{Acknowledgments} We gratefully acknowledge discussions with Sebastiano Carpi, Klaus Fredenhagen, Giuseppe Ruzzi, and Rainer Verch. We are also grateful to Erwin Schr\"odinger Institute, Vienna, for the nice scientific environment that helped much to conclude the paper.

 \appendix
 
  \section{Proof of some propositions} \label{AppA}

\begin{proof} [Proof of Proposition~\ref{proplast}]
 Since all the theory is invariant under
 translation of $\bS^1 = (-\pi, \pi]$ (with $\pi \equiv -\pi$), we can always assume $J_0= (-a,a)$ with $0<|a|< \pi$.
We also select  two other elements $J_1, J_2 \in \cR$ with $(-\pi, \pi) \supset \overline{J_2}$,
$J_2 \supset \overline{J_1}$ and $J_1 \supset \overline{J_0}$.
As a further ingredient we fix an open neighborhood of $1$, $\cO = (e^{-\omega}, e^{\omega})$ with $\omega >0$  so small that
(1) $\lambda \overline{J_0} \subset J_1$, (2) $\lambda \overline{J_2} \subset (-\pi, \pi)$ for all $\lambda \in \cO$. Notice that
$\lambda \in \cO$ iff $\lambda^{-1} \in \cO$.
With these definitions, let $\chi \in C^\infty(\bS^1,\bR)$
 such that $0 \leq \chi(\theta) \leq 1$ for $\theta \in \bS^1$ and, more precisely,
 $\chi(\theta) =1$ for $\theta \in J_1$ but $\chi(\theta) =0$ in $\bS^1 \setminus J_2$.
Now  consider the class of operators $U_\lambda : L^2(\bS^1, d\theta) \to L^2(\bS^1, d\theta)$, with $\lambda \in \cO$, defined by:
  $$(U_\lambda f)(\theta) = \frac{\chi(\theta)}{\sqrt{\lambda}} f(\theta/\lambda)\:, \quad \forall \theta \in (-\pi,\pi]\:.$$
   Using the presence of the smoothing function $\chi$ and using a trivial change of variables where appropriate one proves the following features of $U_\lambda$:
  \begin{eqnarray}
    U_\lambda (C^\infty(\bS^1,\bR))&\subset& C^\infty(\bS^1,\bR)\:, \quad \forall  \lambda \in \cO\:,\label{smoothU}\\
  ||U_\lambda|| &\leq& 1\label{normU}\:, \quad \forall  \lambda \in \cO\:,\\
  U_1\rest_{L^2(J_0, d\theta)} &=& \ide \label{IU}\:,\\
   U_\lambda f &\to& f\quad \mbox{for $\lambda \to 1$ if $f \in C^\infty_0(J; \bC)$\:.} \label{sadded}
  \end{eqnarray}
  By direct inspection one also finds that:
  \beq
  (U^*_\lambda f)(\theta) &=&  \sqrt{\lambda} \chi(\lambda \theta) f(\lambda \theta)\:, \quad \forall f\in
L^2(\bS^1, d\theta)\:, \:\: \forall \theta \in (-\pi,\pi]
\:\:\mbox{and}\:\: \lambda \in \cO\:.\label{U*U}
  \eeq
  Then properties analogous to that found for $U_\lambda$ can be straightforwardly established using the expression given above for $U^*_\lambda$:
   \begin{eqnarray}
    U^*_\lambda (C^\infty(\bS^1,\bR))&\subset& C^\infty(\bS^1,\bR)\:, \quad \forall  \lambda \in \cO\:,\label{smoothU*}\\
  ||U^*_\lambda|| &\leq& 1\label{normU*}\:, \quad \forall  \lambda \in \cO\:,\\
  U^*_1\rest_{L^2(J_0, d\theta)} &=& \ide \label{IU*}\:,\\
   U^*_{1/\lambda} \rest_{L^2(J_0, d\theta)} &=&  U_{\lambda} \rest_{L^2(J_0, d\theta)} \label{UU*}\:, \quad \forall  \lambda \in \cO\:,\\
    U^*_\lambda f &\to& f\quad \mbox{for $\lambda \to 1$ if $f \in C^\infty_0(J)$\:.} \label{ssadded}
  \end{eqnarray}
  \begin{remark} \label{REMARK}
In view of the definition of $U_\lambda$ and (\ref{UU*}), if $f \in C_0^\infty(J_0,\bR)$ then $$\supp (U_\lambda f)= \supp (U^*_{1/\lambda} f)= \lambda \:\supp f\:.$$
\end{remark}

\noindent  Remembering this remark and  taking the first definition in (\ref{cfourier})  into account, one realizes that a candidate for $D_\lambda$ is the operator, initially defined on $C^\infty(\bS^1,\bC)$:
\beq D^{(0)}_\lambda \psi \doteq A^{1/4} U_\lambda A^{-1/4} \: Re \psi \:+\:  i  A^{-1/4} U^*_{1/\lambda} A^{1/4}  \:Im \psi\:,\quad \mbox{for all 
$\lambda \in \cO$ and $\psi \in M_{J_{0}}$.}\label{candidate}\eeq
  The right hand side is in fact  well-defined if $\psi \in K(\cS_{L})$ with $\cR \ni L \subsetneq J_0$, indeed
$A^{-1/4} \: Re \psi$ and $A^{1/4}  \:Im \psi$ belong to $C_0^\infty(J_0, \bR)$ so that they define elements in the domain of
$A^{1/4}$ and $A^{-1/4}$ respectively due to (\ref{smoothU}) and (\ref{smoothU*}).
Moreover it fulfills (a) in the thesis since $D^{(0)}_\lambda \psi \in K(\cS_{\lambda L})$ due to remark \ref{REMARK}.
However both operators $A^{1/4} U_\lambda A^{-1/4}$ and $A^{-1/4} U^*_{1/\lambda} A^{1/4}$ are well defined on $C^\infty(\bS^1,\bC)$.
To extend the validity of (a)  to every space
 $M_{L} \doteq \overline{K(\cS_{L})}$
with $L \subsetneq J_0$ as requested in the thesis,
it is sufficient  to prove that the  operators $A^{1/4} U_\lambda A^{-1/4}$ and $A^{-1/4} U^*_{1/\lambda} A^{1/4}$ are bounded on $C^\infty(\bS^1,\bC)$ and to extend them and $\cD^{(0)}_\lambda$ by continuity on the whole space $L^2(\bS^1, d\theta)$. The restriction
$\cD_\lambda$ to $M_{L}$ of the so obtained continuous extension will satisfy (a) by construction.\\
 To do it we use an argument based on an interpolation theorem.
 Consider $f \in C^\infty(\bS^1,\bC)$ and define $\chi_\lambda(\theta) \doteq \chi(\lambda \theta)$.
     By direct inspection one finds that
     $||A U_\lambda f||_{L^2}^2 \leq \lambda^{-4}||A_{\lambda m} (\chi_\lambda f)||_{L^2}^2$,
  where $A_{\lambda m}$ is $A$ with the mass $m$ replaced by $\lambda m$. By direct inspection one finds also that,
   for $\lambda <1$, $||A_{\lambda m} g||_{L^2}^2$ is bounded by
  $||A g||^2$ otherwise by $\lambda^4 ||A g||^2$. Summarizing
  $$||A U_\lambda f||_{L^2} \leq \: \sup_{\lambda \in \cO}\{1,\lambda^{-4}\}\: ||A (\chi_\lambda f)||^2_{L^2} \:.$$
  We can improve this upper bound as follows expanding $A (\chi_\lambda f)$.
 \beq ||A (\chi_\lambda f)||_{L^2} \leq ||\chi_\lambda A f||_{L^2} + \left|\left|\frac{d^2 \chi_\lambda}{d\theta^2} f \right|\right|_{L^2}
+ 2\left|\left| \frac{d \chi_\lambda}{d\theta} \frac{d f}{d\theta} \right|\right|_{L^2}\:.\label{long}\eeq
    Now, using the expression of the norm and using integration per parts where appropriate:
    \begin{eqnarray}
    ||\chi_\lambda A f||_{L^2} &\leq & ||\chi_\lambda||_\infty ||A f||_{L^2} = ||A f||_{L^2}\:, \nonumber\\
    \left|\left|\frac{d^2 \chi_\lambda}{d\theta^2} f \right|\right|_{L^2} &\leq & \left|\left|\frac{d^2 \chi_\lambda}{d\theta^2}\right|\right|_{\infty}|| f||_{L^2} \:, \nonumber\\
    \left|\left| \frac{d \chi_\lambda}{d\theta} \frac{d f}{d\theta} \right|\right|_{L^2} &\leq &  \left|\left| \frac{d \chi_\lambda}{d\theta}  \right|\right|_{\infty}
    \left|\left| \frac{d f}{d\theta} \right|\right|_{L^2} \leq  \left|\left|\frac{d \chi_\lambda}{d\theta}\right|\right|_{\infty}
   \sqrt{ \left\langle \overline{f} ,\frac{d^2 f}{d\theta^2} \right\rangle} \leq  \left|\left|\frac{d \chi_\lambda}{d\theta}\right|\right|_{\infty} \sqrt{ \left|\left| f \right|\right|_{L^2} \left|\left| \frac{d^2 f}{d\theta^2} \right|\right|_{L^2} }\nonumber
    \end{eqnarray}
  Now notice that  $A \geq \lambda_0 \ide$ where $\lambda_0 >0$ is the least eigenvalue of $A$ (which is strictly positive also for $A_0$) and thus $||Af||_{L^2} \geq \lambda_0 ||f||_{L^2}$. Similarly $A \geq  -\frac{d^2}{d\theta^2}$ and thus $||A f||_{L^2} \geq ||d^2 f/ d\theta^2||_{L^2}$, therefore:
   \begin{eqnarray}
    \left|\left| \frac{d \chi_\lambda}{d\theta} \frac{d f}{d\theta} \right|\right|_{L^2} &\leq &  \lambda_0^{-1/2}\left|\left|\frac{d \chi_\lambda}{d\theta}\right|\right|_{\infty} ||Af||_{L^2} \nonumber
    \end{eqnarray}
  Using these estimates in (\ref{long}) we finally obtains:
\beq ||A U_\lambda f||_{L^2}  \leq C ||Af||_{L^2}\:, \quad \mbox{for all $\lambda \in \cO$ and $f\in C^\infty(\bS^1, d\theta)$,}\label{stimaM}\eeq
 where
 $$C =  \sup_{\lambda \in \cO}\{1,\lambda^{-4}\} \sup_{\lambda \in \cO} \left\{1 + \left|\left|\frac{d^2 \chi_\lambda}{d\theta^2}\right|\right|_{\infty} +
\lambda_0^{-1/2}\left|\left|\frac{d \chi_\lambda}{d\theta}\right|\right|_{\infty}\right\} \:.$$
 $C$ is finite: It can be proved by shrinking $\cO$ and noticing the the function $(\lambda, \theta) \mapsto\chi_\lambda(\theta)$
and its derivatives are bounded in the compact $\overline{\cO} \times \bS^1$ since they are continuous. Since $C^\infty_0(\bS^1, d\theta)$ is a core for the self-adjoint (and thus closed) operator $A$, as a byproduct (\ref{stimaM})implies:
\begin{eqnarray}
U_\lambda (\cD(A)) &\subset& \cD(A)\quad \mbox{for all $\lambda \in \cO$ and} \nonumber\\
||A U_\lambda f||_{L^2}  &\leq& C ||Af||_{L^2}\:, \quad \mbox{for all $\lambda \in \cO$ and $f\in \cD(A)$.} \nonumber
\end{eqnarray}
The proof is immediate noticing that if $f \in \cD(A)$ there is a sequence $C^\infty_0(\bS^1, d\theta) \ni f_n \to f$ with
$Af_n \to Af$ ad , in view of continuity of $U_{\lambda}$,  $\{U_{\lambda}f_n\}_{n\in \bN}$ is Cauchy and, in view of (\ref{stimaM}) $\{AU_{\lambda}f_n\}_{n\in \bN}$
is Cauchy too. Closedness of $A$ implies that $Uf_n \to Uf \in \cD(A)$ and $A(Uf_n) \to A(Uf)$. This also proves that (\ref{stimaM})
is still valid in $\cD(A)$ by continuity. As $A\geq 0$ and (\ref{normU}) is valid, Proposition 9 cap IX.5 in Reed-Simon vol.2
used twice implies that 
 \begin{eqnarray}
U_\lambda (\cD(A^{1/4})) &\subset& \cD(A^{1/4})\quad \mbox{for all $\lambda \in \cO$ and} \nonumber\\
||A^{1/4} U_\lambda f||_{L^2}  &\leq& C^{1/4} ||A^{1/4}f||_{L^2}\:, \quad \mbox{for all $\lambda \in \cO$
and $f\in \cD(A^{1/4})$,}\nonumber
\end{eqnarray}
so that, since $Ran(A^{-1/4}) = \cD(A^{1/4})$ and
$\cD(A^{-1/4})$ is the whole Hilbert space, in particular
\beq A^{1/4} U_\lambda A^{-1/4}= B_\lambda: L^2(\bS^1,d\theta)\to L^2(\bS^1,d\theta) \quad  \mbox{with  $||B_\lambda|| \leq C^{1/4}$  for all $\lambda \in \cO$.} \label{stime}\eeq
This concludes the proof of the continuity of the
former operator in the right-hand side of (\ref{candidate}). Let us focus on the latter operator.
 By construction we obtain on the dense domain $\cD(A^{1/4})$, taking the adjoint of $B_\lambda$ and replacing $\lambda$ with $1/\lambda$ (remind that $\lambda \in \cO$ iff $1/\lambda \in \cO$):
$A^{-1/4} U^*_{1/\lambda} A^{1/4} \subset  B_{1/\lambda}^*$.
Since $B^*_{1/\lambda}$ is defined on the whole Hilbert space and $||B^*_{1/\lambda}|| = ||B_{1/\lambda}|| \leq C^{1/4}$, it being the adjoint of a bounded every-here defined operator, we conclude that
\beq A^{-1/4} U^*_{1/\lambda} A^{1/4} \:\mbox{cont. extends to} \:\: B^*_{1/\lambda}: L^2(\bS^1,d\theta)\to L^2(\bS^1,d\theta) \:\: \mbox{with  $||B^*_{1/\lambda}|| \leq C^{1/4} \forall \lambda \in \cO$.} \label{stime2}\eeq
 This concludes the proof of (a).\\
 Concerning the property (b): $D_\lambda \psi \to \psi$
  as $\lambda \to 1$ for $\psi \in M_L$ with $\cR \ni L \subsetneq J_0$, it is equivalent to prove that $B_\lambda Re \psi \to Re \psi$ and $B^*_{1/\lambda} Im \psi \to Im \psi$
  as $\lambda \to 1$ for $\psi \in M_L$.\\
  Notice that $A^{-1/4}$ is continuous and so, when $\psi \in K(\cS_L)$, one has
   $$A^{-1/4} U^*_{1/\lambda} A^{1/4} (Im \psi) \to A^{-1/4} U^*_1 A^{1/4} (Im \psi) =  A^{-1/4} A^{1/4} (Im \psi) = Im \psi$$
where we have used  (\ref{ssadded}) and
(\ref{IU*}) noticing that $A^{1/4} (Im \psi) \in C^\infty(L) \subset  L^2(J_0, d\theta)$ when $\psi \in K(\cS_L)$.
 The result can be extended to $M_{L}\doteq \overline{K(\cS_L)}$ due to the uniform bound (\ref{stime}) as follows. If $\psi\in M_{L}$, let
   $K(\cS_L) \ni \psi_n \to \psi$ and denote $Im \psi_n$ and $Im \psi$ respectively by $f_n$ and $f$. Obviously $f_n \to f$. One has, for $\lambda \in \cO$ so that (\ref{stime}) holds,
   $$||B^*_{1/\lambda} f- f|| \leq ||B^*_{1/\lambda} (f-f_n)|| + ||B^*_{1/\lambda} f_n -f_n|| + ||f_n-f|| \leq (C^{1/4}+1)||f-f_n|| + ||B^*_{1/\lambda} f_n -f_n||\:.$$
   For any fixed $\epsilon>0$, taking $n= n_\epsilon$  such that $(C^{1/4}+1)||f-f_{n_\epsilon}|| < \epsilon/2$, we can found $\delta>0$
  such that $\lambda \in (1-\delta, 1+\delta)$ entails $||B^*_{1/\lambda}f_{n_\epsilon} -f_{n_\epsilon}||< \epsilon/2$. Hence
  for that $\epsilon>0$, $||B^*_{1/\lambda} f- f|| <
  \epsilon$  provided  that $\lambda \in (1-\delta, 1+\delta)$. That is $B^*_{1/\lambda} Im \psi \to Im \psi$ as $\lambda \to 1^-$ for all $\psi \in M_{L}$. \\
  To conclude let us pass to prove that   
$B_{\lambda} Re \psi \to Re \psi$
  as $\lambda \to 1$ for $\psi \in M_L$ and $L \subsetneq J_0$. Let us indicate $Re\psi$ by $f$. As before, first consider the case
$\psi \in K(\cS_L)$. This means in particular that  $f = A^{1/4}h$ for some $h\in C_0^\infty(J_0, \bR)$. Now notice that:
\beq
||B_{\lambda} f - f||^2 = ||B_{\lambda} f||^2 + ||f||^2 - 2 Re \langle f, B_{\lambda} f\rangle \:.\label{ustep}\eeq
In our case, as $\lambda \to 1^-$, due to (\ref{IU}) and (\ref{sadded}):
$$\langle f, B_{\lambda} f\rangle = \langle A^{1/4} h, A^{1/4} U_\lambda h\rangle = \langle A^{1/2} h, U_\lambda h\rangle \to
\langle A^{1/2} h,  h\rangle =  \langle A^{1/4} h, A^{1/4} h\rangle = \langle  f,  f\rangle\:, $$ 
Similarly $||B_{\lambda} f||^2 \to \langle  f,  f\rangle$ as $\lambda \to 1^-$, this because:
$$||B_{\lambda} f||^2 = \langle f, A^{1/4} U^*_\lambda A^{-1/4} A^{1/4} U_\lambda A^{-1/4} f\rangle =
 \langle f, A^{1/4} U^*_\lambda U_\lambda A^{-1/4} A^{1/4} h\rangle =  \langle f, A^{1/4} U^*_\lambda U_\lambda h\rangle$$
and by direct inspection, using the definition of $U_\lambda$ and (\ref{U*U}) one see that, for each $h\in C_0^\infty(J_0, \bR))$
$U^*_\lambda U_\lambda h = h$. Putting all together in (\ref{IU}) one concludes that
$B_{\lambda} Re \psi \to Re \psi $ as $\lambda \to 1^-$ when $\psi \in K(\cS_L)$. The extension to the case $\psi \in M_L \doteq \overline{K(\cS_L)}$ is the same as in the case of $Im \psi$.\\
  We have proved the property (b) that $D_\lambda \psi \to \psi$ as $\lambda \to 1^-$ for $\psi \in M_{L}$
  and it concludes the proof.
\end{proof}

\begin{proof}[Proof of Lemma~\ref{lemmaadded}]
 Notice that, as a general fact it  holds $\cS_I \subset (\cS_{I'})'= (M_{I'})'$ and
$\cS_J \subset (\cS_{J'})'= (M_{J'})'$
and thus taking the closures and the intersections, $M_I \cap M_J  \subset (M_{I'})' \cap (M_{J'})'$.
 This is equivalent to say that, if
$\psi \in M_I \cap M_J$ then
$Im \langle \psi, \phi \rangle =0$ when either $\phi\in M_{I'}$ or $\phi \in M_{J'}$. In particular,
$Im \langle \psi, K(\Phi,\Pi) \rangle =0$ when both the smooth real functions $\Phi, \Pi$ are supported in
$I'$ or in $J'$. Therefore the distributions (see the proof of Theorem \ref{RSp} to show that
that those functionals are in fact distributions of $\cD'(\bS^1)$) $C^\infty(\bS^1,\bR) \ni
f \mapsto \langle Im \psi, A^{1/4} f\rangle$
and $C^\infty(\bS^1,\bR) \ni
f \mapsto \langle Re \psi, A^{-1/4} f\rangle$ have support included in $\bS^1 \setminus (I'\cup J') = (\bS^1 \setminus I') \cap (\bS^1 \setminus J')
= \overline{I} \cap \overline{J}$. Since $I$ and $J$ are disjoint proper open segments one has
$\overline{I} \cap \overline{J}=\partial I \cap \partial J$.
Therefore, if $\partial I \cap \partial J= \emptyset$ both distributions $\langle Re \psi, A^{-1/4} \cdot \rangle$
and $\langle Im \psi, A^{1/4}  \cdot\rangle$ vanish and this implies that $\psi =0$ since
$\overline{A^{\pm 1/4}(C^\infty(\bS^1))} = L^2 (\bS^1,d\theta)$ as proved in Proposition \ref{P1}. Otherwise
$\partial I \cap \partial J$ contains two points at most, say $p$ and $q$.
We can assume, without loss of generality, that $\theta_p=0$ and $\theta_q \in (0,2\pi)$ (this extend can always be achieved by
redefining the origin of
coordinate $\theta$ on $\bS^1$). It is  a well-known result
of distributions theory that distributions with support given by a single point are
polynomials of derivative of Dirac deltas supported on that point (the case of a finite number of points is a
trivial extension). Consider $\langle Im \psi, A^{1/4}  f\cdot\rangle$. In our case there must be a finite number
of coefficients $a_j,b_j \in \bR$ such that, for every $f \in C^\infty(\bS^1,\bR)$  it must hold
$$\langle Im \psi, A^{1/4}  f \rangle = \sum_{j=0}^{N_p} a_j \frac{d^j}{d\theta^j} f|_{p} +
\sum_{j=0}^{N_q} b_j \frac{d^j}{d\theta^j} f|_{q}\:.$$
Passing to Fourier transformation, the identity above can be re-written if $\psi_k$ and $f_k$ are the Fourier
coefficients of $Im \psi$ and $f$ respectively
$$\sum_{k\in \bZ}  \overline{\psi_k} (k^2+m^2)^{1/4} f_k = \sum_{k\in \bZ}  \left(
\sum_{j=0}^{N_p} a_j (ik)^j + \sum_{j=0}^{N_q} b_j (ik)^j e^{ik \theta_q}\right) f_k\:.$$
(notice that $f_k \to 0$ faster than every power $|k|^{-M}$ so that the right hand side is well defined).
Since the functions $f$ are dense in the Hilbert space, this is equivalent to say that:
\begin{equation}\label{AGG}
 \overline{\psi}_{k}\doteq (k^{2}+m^{2})^{-\frac{1}{4}}\left(\sum_{j=0}^{N} (a_{j} + e^{\imath k\theta_{q}} b_{j})(i k)^{j}\right)\:.
\end{equation}
where we have defined $N\doteq\max(N_{p},N_{q})$ (assuming $a_j=0$ and $b_j=0$ for the added coefficients).
Let us prove that the right-hand side defines a $\ell^2(\bZ)$ sequence -- as it is required by $\psi \in L^2(\bS^2,d\theta)$ --  only if $a_j=0$ and $b_j=0$ for every $j$.
Assume that $\{{\psi}_{k}\}_k \in \ell^2(\bZ)$ so that the right-hand side of (\ref{AGG}) defines a $\ell^2(\bZ)$ sequence.
If $c_{j,k}\doteq Re\left[ (a_{j} + e^{i k\theta_{q}} b_{j})i^{j}\right]$,
\begin{equation}\label{ReRaw}
(Re\overline{\psi}_{k})^{2}= \left( k^{2}+m^{2}\right)^{-\frac{1}{2}}\sum_{l,j=0}^{N} c_{j,k}\,c_{l,k}\, k^{l+j}\:.
\end{equation}
The sequence $\{k\theta_{q}\}_{k\in\mathbb{Z}}$ in $[0, 2\pi]$ may be either
 periodic -- and this happens when $\frac{\theta_{q}}{2\pi}$ is rational --
 or it is dense in $[0, 2\pi]$ -- and this arises for $\frac{\theta_{q}}{2\pi}$ irrational. Fix
 $k_{0}\in\mathbb{Z}\setminus \{0\}$, in both
cases for $\epsilon > 0$, there is a sequence of integers
$\{k^{(\epsilon)}_{n}\}_{n\in\mathbb{Z}}$ such that:
\begin{equation*}             |c_{N,k^{(\epsilon)}_{n}}-c_{N,k_{0}}|<\epsilon\:,\quad\forall\; n\in\mathbb{Z}\:.
\end{equation*}
Moreover, defining $M \doteq \max_{j=0,\ldots, N} |a_j| +|b_j|$ one has $c_{j,k} \geq -M> -\infty$,
therefore a lower bound for the right-hand side of (\ref{ReRaw}) is
\begin{equation}\label{ReIneq}
(Re \overline{\psi}_{k^{(\epsilon)}_{n}})^{2}\geq ((k^{(\epsilon)}_{n})^{2}+m^{2})^{-\frac{1}{2}}\left((c_{N,k_{0}}-
\mbox{sign}(c_{N,k_{0}})  \epsilon)^{2}\,
(k_{n}^{(\epsilon)})^{2N} -
 \sum_{l+j<2N} M^2\,|k^{(\epsilon)}_{n}|^{l+j}
\right),
\end{equation}
If $c_{N,k_{0}} \neq 0$ the leading term in (\ref{ReIneq}) is
$((k_{n}^{(\epsilon)})^{2}+m^{2})^{-\frac{1}{2}}(c_{N,k_{0}}-\mbox{sign}(c_{N,k_{0}}) \epsilon)^{2}(k_{n}^{(\epsilon)})^{2N}$, so that
the right-hand side of (\ref{ReIneq})
diverges to $+\infty$ -- and
 $\{{\psi}_{k}\}_{k}\notin\ell^{2}(\mathbb{Z})$ -- unless $c_{N,k_{0}}-\mbox{sign}(c_{N,k_{0}}) \epsilon=0$.
Arbitrariness of  $\epsilon$ implies $c_{N,k_{0}}=0$ that is
$Re \left[ (a_{N} + e^{i k_{0}\theta_{q}} b_{N})i^{j}\right]= 0$.
Analogously one sees that
$Im \left[ (a_{N} + e^{i k_{0}\theta_{q}} b_{N})i^{j}\right] = 0
$,
and thus
$a_{N} + e^{i k_{0}\theta_{q}} b_{N}=0$. However,
since  $k_{0}$ was arbitary one also has  $a_{N}=b_{N}=0$.
Iterating the procedure one achieves
$
a_{j}=b_{j}=0,\quad \forall\; j=N,N-1,\ldots,1\:.
$
So that it remains to consider the case $j=0$, that is the case of $\{\psi_k\}_{k\in \bZ} \in \ell^2(\bZ)$ with
\begin{equation*}
{\overline{\psi}}_{k}\doteq (k^{2}+m^{2})^{-\frac{1}{4}}\left(a_{0} + e^{i k\theta_{q}} b_{0}\right)\quad \mbox{where $a_0,b_0\in \bR$ are
constant}\:.
\end{equation*}
Now
\begin{equation*}
|{\psi}_{k}|^{2} =
|\overline{\psi}_{k}|^{2}=(k^{2}+m^{2})^{-\frac{1}{2}}|a_{0} +
 e^{i k\theta_{q}} b_{0}|^{2}\geq (k^{2}+m^{2})^{-\frac{1}{2}}||a_{0}|-|b_{0}||^{2},                                                           
\end{equation*}
and thus $\{\psi\}_{k}\notin\ell^{2}(\mathbb{Z})$ unless $b_{0}=\pm a_{0}$.
With that choice we have in turn:
\begin{equation*}
|{\psi}_{k}|^{2}=2a_{0}^{2}\,(k^{2}+m^{2})^{-\frac{1}{2}}(1\pm\cos(k\theta_{q})).                             
\end{equation*}
As the series $\sum_{k=0}^{\infty}\frac{\cos(k\theta_{q})}{k}$ converges ($\theta_{q}\neq 0 \mod{2\pi}$ by hypotheses),
and $(k^{2}+m^{2})^{-\frac{1}{2}} \sim \frac{1}{k}$ for $k\rightarrow\infty$, it arises that
 $\sum_{k=0}^{\infty}|{\psi}_{k}|^{2}$ diverges barring the case $a_{0}= b_0 = 0$. 
This concludes the proof of the fact that  $a_j = b_j=0$ for all $j$ if $\psi \in L^2(\bS^1, d\theta)$.
We have found that the distribution  $\langle Im \psi, A^{1/4} \cdot \rangle$ must vanish.
The proof for $\langle Re \psi, A^{-1/4}  \cdot \rangle$ is strictly analogous. Since both distributions vanish
and $\overline{A^{\pm 1/4}(C^\infty(\bS^1))} = L^2 (\bS^1,d\theta)$, we are commited to admit that $\psi =0$, so that
$M_I \cap M_J = \{0\}$. \\
Concerning the last statement, from (\ref{LRT2}) one has
$\gR(I) \cap \gR(J) =  \gR[M_I] \cap \gR[M_J] = \gR[M_I\cap M_J] = \gR[\{0\}] = \bC \ide$.
\end{proof}

\begin{proof}[Proof of Lemma \ref{lemmaE}]
In the following  $\lambda \in \bR$. We complete the unit-norm vector $\psi\in \cH$ to a Hilbert basis of
$\cH$, pass to the associated Hilbert basis in $\gF_+(\cH)$
 and denote by  $F$ the dense subspace of $\gF_+(\cH)$ contianing all af finite linar combinations of the vectors of that basis.
Assuming $E_{\lambda}^{(\psi)}= e^{i\lambda A}$, taking the derivative at $\lambda=0$ of the identity
$$E_{\lambda}^{(\psi)}  W[\psi] E_{\lambda}^{(\psi)*}  =
W\left[e^\lambda \psi\right]\:, 
$$
without paying much attention to domain issues and, finally,  making use of (\ref{weylgen}), one gets that
\beq
\left[iA, a(\psi) -a^*(\psi)\right] \Phi = (a(\psi)-a^*(\psi)) \Phi \:,\label{commF}
\eeq
if $\Phi$ belongs to some suitable domain we  shall determine shortly.
Taking the commutation relation
$[a(\psi), a^*(\psi)] = \ide$  into account (recall that $||\psi||=1$),  we see that a candidate for $A$
is some self-adjoint extension of $A\doteq (1/2) (i a(\psi)a(\psi) - i a^*(\psi)a^*(\psi))$.
$A$ turns out to be symmetric if defined on $F$. If $\Phi \in F$ contains exactly $k$ particles in the state $\psi$
one finds $||A^n \Phi|| \leq \sqrt{(2n + k)!}$. From that it arises
 $\sum_{n=0}^{+\infty}\lambda^n ||A^n \Phi||/n! <+\infty$ if $|\lambda| <1/2$. Therefore
the vectors of $F$ are analytic for $A$ and thus
 $A$ is essentially
self-adjoint on $F$, $\overline{A}$ being its unique self-adjoint extension.
In particular the commutation relation (\ref{commF}) are, in fact, valid for $\Phi \in F$.
 Relations (\ref{commF}) lead to the further commutation
 relations
\beq
\left[(iA)^n, a(\psi) -a^*(\psi)\right] \Phi = \sum_{k=0}^{n-1} \binom{n}{k}  (a(\psi)-a^*(\psi)) (iA)^k \Phi \quad \mbox{for all $\Phi
\in F$.} \label{commF2}
\eeq
Using (\ref{commF2}) one easily proves the validity of the identity  for $|\lambda|<1/4$ and $\Phi \in F$,
\beq \sum_{n=0}^{+\infty} \frac{(i\lambda A)^n}{n!} (a(\psi)-a^*(\psi)) \Phi  =
 \sum_{n=0}^{+\infty} (a(\psi)-a^*(\psi)) \frac{(i\lambda A + \lambda \ide)^n}{n!}  \Phi\:.\label{InterE}\eeq
The series
$\sum_{n=0}^{+\infty} \frac{(i\lambda A + \lambda \ide)^n}{n!}  \Phi$
 converges for every $\Phi \in F$ and $|\lambda |< 1/4$ as one can establish making use of the bounds
 $||(A+i\ide)^n \Phi|| \leq 2^n \sqrt{(2n+k)!}$ when $\Phi \in F$ contains exactly $k$ particles in the state $\psi$. Therefore closedness of $\overline{a(\psi)-a^*(\psi)}$
 imply, via (\ref{InterE}), that the following two facts hold
 (i) $\sum_{n=0}^{+\infty} \frac{(i\lambda A + \lambda \ide)^n}{n!}  \Phi \in \cD(\overline{a(\psi)-a^*(\psi)})$ when
 $\Phi \in F$, $|\lambda|< 1/4$ and
 (ii) $\overline{a(\psi)-a^*(\psi)} \sum_{n=0}^{+\infty} \frac{(i\lambda A + \lambda \ide)^n}{n!} 
 \Phi = \sum_{n=0}^{+\infty} (a(\psi)-a^*(\psi)) \frac{(i\lambda A + \lambda \ide)^n}{n!}  \Phi$
in the same case. Therefore (\ref{InterE}) can be re-written as
\beq e^{i\lambda \overline{A}} (a(\psi)-a^*(\psi)) \Phi  =
\overline{a(\psi)-a^*(\psi)} \sum_{n=0}^{+\infty}  \frac{(i\lambda A + \lambda \ide)^n}{n!}  \Phi\:,\label{InterEx}\eeq
 where we have also used the fact that $(a(\psi)-a^*(\psi)) \Phi \in F$ when $\Phi\in F$ and thus the exponential
 $e^{i\lambda \overline{A}} (a(\psi)-a^*(\psi)) \Phi$ can be expanded in series. Since $\lambda \ide$ and $i\lambda A$
 commute, following exactly the same proof as used for numbers,
 one achieves $ \sum_{n=0}^{+\infty}  \frac{(i\lambda A + \lambda
 \ide)^n}{n!}  \Phi =e^\lambda \sum_{n=0}^{+\infty} 
 \frac{(i\lambda A)^n}{n!}\Phi$. On the other hand, since $\Phi$ is analytic for $A$, the right-hand side is nothing but
 $e^{\lambda} e^{i\lambda \overline{A}} \Phi$.   
Summing up, the identity (\ref{InterEx}) can be re-stated as
$$e^{i\lambda\overline{A}} \overline{a(\psi)-a^*(\psi)} \Phi =  \overline{e^\lambda(a(\psi)-a^*(\psi))} e^{i\lambda\overline{A}}\Phi\:,
\quad \mbox{for every $\Phi \in F$ and $|\lambda| <1/4$.}$$
This identity, used recorsively, leads immediately to
\beq e^{i\lambda\overline{A}} \overline{a(\psi)-a^*(\psi)}^n \Phi = \overline{e^{\lambda}(a(\psi)-a^*(\psi))}^n
e^{i\lambda\overline{A}}\Phi \label{InterE2}\quad \mbox{for every $\Phi \in F$ and $|\lambda| <1/4$.}\eeq
Since $e^{i\lambda\overline{A}}$ is unitary, (\ref{InterE2}) entails that, for $\Phi \in F$, $|\lambda| <1/4$ and every $u \in \bC$:
$$ \sum_{n=0}^\infty \frac{u^n}{n!} \left|\left| \overline{e^{\lambda}(a(\psi)-a^*(\psi))}^n e^{i\lambda\overline{A}} \Phi \right|
\right| =
 \sum_{n=0}^\infty \frac{u^n}{n!}\left|\left| \overline{a(\psi)-a^*(\psi)}^n \Phi\right| \right| < +\infty\:,
$$
where we have used the fact that every $\Phi \in F$ is analytic (for every value of the parameter $u$) for $ \overline{ia(\psi)-ia^*(\psi)}$
as is well known (see \cite{BRII}). We have found that $e^{i\lambda\overline{A}} \Phi$ is analytic for
$\overline{e^{\lambda}(a(\psi)-a^*(\psi))}$. In this context, the identity arising from (\ref{InterE2}) for
$\Phi \in F$ and $|\lambda| <1/4$,
$$e^{i\lambda\overline{A}} \sum_{n=0}^{+\infty}\frac{1}{n!}\overline{a(\psi)-a^*(\psi)}^n \Phi =
\sum_{n=0}^{+\infty}\frac{1}{n!} \overline{e^{\lambda}(a(\psi)-a^*(\psi))}^n
e^{i\lambda\overline{A}}\Phi $$
 can be re-written
$e^{i\lambda\overline{A}} e^{\overline{a(\psi)-a^*(\psi)}} \Phi =
e^{\overline{e^\lambda (a(\psi)-a^*(\psi))}} 
e^{i\lambda\overline{A}} \Phi$.
That is, taking advantage from the fact that $F$ is dense,
$E_{\lambda}^{(\psi)}  e^{\overline{a(\psi)-a^*(\psi)}} E_{\lambda}^{(\psi)*}  =
e^{ \overline{e^\lambda(a(\psi)-a^*(\psi))}} 
$,
where we have defined $E_{\lambda}^{(\psi)} \doteq e^{i\lambda\overline{A}}$.
Finally, employing $\bR$-linearity of $\psi \mapsto a(\psi), a^*(\psi)$, the achieved formula can be
re-stated as
$$E_{\lambda}^{(\psi)}  W[\psi] E_{\lambda}^{(\psi)*}  =
W\left[e^\lambda \psi\right]\:. 
$$
The restriction $|\lambda| <1/4$ can be dropped by employing iteratively the identity above and noticing that
$E_{\lambda}^{(\psi)}$ is additive in $\lambda\in \bR$. Hence the obtained identity holds true for every
$\lambda \in \bR$.
\end{proof}

 \section{Universal algebras}\label{AppB}
Consider a class of $C^*$-algebras with unit $\ide$ in common,
 $\{\cA(I)\}_{I\in \cI}$, where $\cI$
is a partially ordered set. We denote by $\subset$ the ordering relation in
$\cI$. Assume that the class $\{\cA(I)\}_{I\in \cI}$ is isotonous, i.e.
$$\cA(I) \subset \cA(J) \quad \mbox{when $I\subset J$ for $I,J\in \cI$}$$
where $\cA(I) \subset \cA(J)$ means that the former is a sub $C^*$-algebra
of the latter, these requirements define a {\em pre-cosheaf} of $C^*$-algebras (see for instance \cite{GL92}). It is {\em not} assumed that  $\{\cA(I)\}_{I\in \cI}$ is directed
with respect to $\subset$ and thus one cannot define the inductive limit of the
class $\cA$. However, as pointed out by Fredenhagen in \cite{Fredenhagen} (see also \cite{GL92} for a different but equivalent approach), it is possible to give a sort of generalized  inductive limit
of the isotonous
class of $C^*$-algebras $\{\cA(I)\}_{I\in \cI}$ which corresponds, in physical
application, to the $C^*$-algebra of quasi local observables also in those
contexts where the set $\cI$ is not directed. This is the case treated in this
paper when $\cI= \cR$ and $\cA(I)= \cW(I)$.\\

\begin{definition} \label{UA} Consider a {\em pre-cosheaf} of $C^*$-algebras with unit $\ide$ in common,
 $\{\cA(I)\}_{I\in \cI}$, referred to the partially ordered set $(\cI, \subset)$.\\
A $C^*$-algebra with unit $\cA$ is called an {\bf universal algebra} associated with $\{\cA(I)\}_{I\in \cI}$
if it fulfills the following properties.
\begin{itemize}
\item[$(1)$] $\cA$ contains every $\cA(I)$ as a $C^*$-subalgebra for $I \in \cI$ and coincides with the $C^*$-algebra generated by all of the
subalgebras together\footnote{This requirement was not assumed in \cite{Fredenhagen}
 but it has been added in the subsequent \cite{FRS}. It is essential for the uniqueness of $\cA$.}\:,

\item[$(2)$] if $\{\pi_I\}_{I\in \cI}$ is a class of representations on $\cB(\cH)$, for some Hilbert space $\cH$:
$$\pi_I : \cA(I) \to \cB(\cH)\:,$$
satisfying  compatibility conditions
\beq \pi_I\rest_{\cA(J)} = \pi_J\quad \mbox{when $J\subset I$ for $I,J\in \cI$}\:,\label{comp}\eeq
then there is a {\em unique} representation $\pi: \cA \to \cB(\cH)$ such that:
\beq \pi\rest_{\cA(I)} = \pi_I\quad \mbox{for every $I\in \cI$}\:. \label{comp2}\eeq
\end{itemize}
\end{definition}

\begin{proposition}
 Consider a {\em pre-cosheaf} of $C^*$-algebras with unit $\ide$ in common,
 $\{\cA(I)\}_{I\in \cI}$, referred to the partially ordered set $(\cI, \subset)$. 
  The following facts hold.
\begin{itemize}
\item[$(a)$] If the free $*$-algebra generated by the algebras $\cA(I)$, $I \in \cI$ admits a non-trivial $C^*$-algebra 
 seminorm,  $\{\cA(I)\}_{I\in \cI}$ admits a universal algebra $\cA$.
\item[$(b)$] If the universal algebra exists, it is uniquely determined 
 up to $C^*$-algebra isomorphisms.
\item[$(c)$] If $(\cI, \subset )$ is directed, $\cA$ exists and it is isomorphic to the inductive limit of $\{\cA(I)\}_{I\in \cI}$.
\end{itemize}
\end{proposition}

\begin{proof}
$(a)$ The existence of a universal algebra $\cA$ has been proved in \cite{Fredenhagen}
assuming the existence
of a nontrivial $*$-representation of the $*$-algebra freely generated by the
precosheaf satisfying the compatibility conditions (\ref{comp}) and (\ref{comp2}). This
hypothesis is tantamount to the existence of a nontrivial $C*$-seminorm
on this $*$-algebra, and has been checked in examples such as those of
chiral conformal models on the circle (see Fredenhagen's aforementioned
reference for more details), but its validity in general is unknown.\\
$(b)$ Consider two universal algebras $\cA_1$ and $\cA_2$ and (faithfully and isometrically) represent these $C^*$-algebras
in terms of subalgebras of $\cB(\cH_1)$ and $\cB(\cH_2)$ respectively, for suitable Hilbert spaces $\cH_1$ and $\cH_2$.
For $i=1,2$ the classes of embeddings $\{(\imath_I)_i\}_{I\in\cI}$
 $(\imath_I)_i : \cA(I) \to \cA_i$ can be viewed as classes of representations $\{(\pi_I)_i\}_{I\in\cI}$
valued on $\cB(\cH_i)$. By construction both $\{(\pi_I)_1\}_{I\in\cI}$ and $\{(\pi_I)_2\}_{I\in\cI}$
fulfill separately the compatibility conditions (\ref{comp}). Considering $\cA_1$ as the universal algebra, property (2)
of the definition implies that there is  representation $\pi_{12}: \cA_1 \to \cB(\cH_2)$ such that
$$\pi_{12} \circ (\pi_I)_{1} = (\pi_I)_{2}\quad \forall I \in \cI\:.$$
Interchanging the role of $\cA_1$ and $\cA_2$, one finds another representation $\pi_{21}:\cA_2 \to \cB(\cH_1)$ with
 $$\pi_{21} \circ (\pi_I)_{2} = (\pi_I)_{1}\quad \forall I \in \cI\:.$$
 These two classes of identities together implies:
 $$(\pi_{21} \circ \pi_{12})\rest_{\pi_I(\cA(I))} = id_{(\pi_I)_1(\cA(I))}\:, \quad (\pi_{12} \circ \pi_{21})\rest_{(\pi_I)_2(\cA(I))} =
 id_{(\pi_I)_2(\cA(I))} \quad \forall I \in \cI\:.$$
 Then using continuity of representations $\pi_{21}$
and  $\pi_{12}$ and closedness of their domains,
the identities above entail
that (i) $\pi_{21}$ includes $\pi_{12}(\cA_{g1})$ in its  domain
and  $\pi_{12}$ includes $\pi_{21}(\cA_{g2})$ in its  domain,
where $\cA_{g1}$ and $\cA_{g2}$ are the sub $C^*$-algebras of $\cA_1$ and $\cA_2$ respectively generated by all of
$\cA_1(I)$ and all of $\cA_2(I)$, and (ii)
${\pi_{21}} \circ \pi_{12}\rest_{\cA_{g1}}  = id_{\cA_{g1}}\:, \quad {\pi_{12}} \circ \pi_{21}
\rest_{\cA_{g2}}  =
 id_{\cA_{g2}}$.
 Since $\cA_{gi} = \cA_i$ we have actually obtained that:
 $${\pi_{21}} \circ \pi_{12} = id_{\cA_{1}}\:,
 \quad {\pi_{12}} \circ \pi_{21} =
 id_{\cA_{2}} $$
so that $\pi_{12}$ and $\pi_{21}$ are in fact $C^*$-algebra isomorphisms, and, in particular $\cA_2 = \pi_{12}(\cA_1)$.\\
(c) The inductive limit $\cA$ is the completion of the $*$-algebra $\bigcup_{I\in \cI} \cA(I)$.
If $a\in \cA$,
there must be a sequence $\{I_n\}_{n\in \bN} \subset \cI$, with $I_i \subset I_k$  for $i \leq  k$,
such that $a_n \to a$ as $n\to +\infty$ and $a_n \in \cA(I_n)$. if $\{\pi_I\}_{I\in \cI}$ is a class
of representations on $\cB(\cH)$, for some Hilbert space $\cH$:
$$\pi_I : \cA_I \to \cB(\cH)\:,$$
satisfying  compatibly conditions (\ref{comp})
and $\pi$ is a representation (on $\cB(\cH)$) of $\cA$ which reduces to $\pi_I$ on every $\cA(I)$, it holds, {\em remembering that
representations are norm decreasing and thus continuous}:
$$\pi(a) = \pi\left(\lim_{n \to +\infty} a_n\right) = \lim_{n \to +\infty} \pi(a_n) = \lim_{n \to +\infty} \pi_{I_n}(a_n)$$
 so that $\pi$ is completely individuated by the class of $\pi_I$. On the other hand, such a class of representations
individuates a representation $\pi$ of $\cA$ by means of the same rule
(notice that, if $m\geq n$,
 $||\pi_{I_n}(a_n) - \pi_{I_m}(a_m)|| = ||\pi_{I_m}(a_n)- \pi_{I_m}(a_m)|| \leq ||a_n-a_m\||$
 so that $\{\pi_n(a_n)\}$ is Cauchy when $\{a_n\}$ is such). We have proved that the inductive limit is a universal algebra.
 \end{proof}

 \begin{remark} If $B$ is a sub unital $C^*$-algebra of a unital $C^*$-algebra $A$ and every representation $\pi$ of $B$ on some space of
 bounded operators on ha Hilbert space $\cB(\cH)$ admits a unique extension to $A$, it is anyway possible that $B \subsetneq A$
 (it is sufficient that $B$ includes a closed two-sided ideal of $A$, see Dixmier book). Therefore the requirement that the sub
 algebras $\cA(I)$ generates $\cA$ is essential in proving the uniqueness of the universal algebra $\cA$.
 \end{remark}

\noindent As an example consider the theory on $\bS^1$ studied in the paper and focus on the class of unital
$C^*$ algebras (Weyl algebras) $\{\cW(I)\}_{I \in \cR}$. It is simply proved that $\cW$ is the associated universal algebras. \\

\begin{proposition} $\cW$ is the  universal algebra for 
$\{\cW(I)\}_{I \in \cR}$.
\end{proposition}

\begin{proof}
Condition $(1)$ in definition \ref{UA} is trivially fulfilled.
Then consider a class of representations $\{\pi_I\}_{I\in \cI}$ on $\cB(\cH)$, for some Hilbert space $\cH$
satisfying  compatibility conditions (\ref{comp}). Suppose that there is $\pi: \cW \to \cB(\cH)$ satisfying (\ref{comp2}).
Fix $I,J \in \cR$ with $I \cup J = \bS^1$ and $f,g \in C^\infty(\bS^1,\bR)$
with $f+g =1$ and $\supp f \subset I$, $\supp g \subset J $
 For $(\Pi,\Phi) \in \cS$ one has, if  $h(\Phi,\Pi)$ denotes the couple $(h\cdot \Phi, h \cdot \Pi)$:
$$
\pi\left(W(\Phi,\Pi)\right) = \pi\left(W(f(\Phi,\Pi) + g(\Phi,\Pi))\right) = \pi\left(W(f(\Phi,\Pi)\right)
 \pi\left(W(g(\Phi,\Pi))\right) e^{-i\sigma(f(\Phi,\Pi), g(\Phi,\Pi))/2}\:.$$
 We have found that:
$$\pi\left(W(\Phi,\Pi)\right) =e^{-i\sigma(f(\Phi,\Pi), g(\Phi,\Pi))/2} \pi_I\left(W(f(\Phi,\Pi)\right)
 \pi_J\left(g(\Phi,\Pi))\right)
$$
Incidentally, by direct inspection, one finds that $\sigma(f(\Phi,\Pi), g(\Phi,\Pi))=0$  also if $f\cdot g \neq 0$. Therefore
\beq \pi\left(W(\Phi,\Pi)\right) = \pi_I\left(W(f(\Phi,\Pi)\right)
 \pi_J\left(g(\Phi,\Pi))\right) \:. \label{coinc}
\eeq
The right-hand side does not depend on $\pi$. Since every element of $\cW$ is obtained by linearity and continuity
from generators $W(\Phi,\Pi)$ and representations are continuous, we conclude that every representation of $\cW$ satisfying
(\ref{comp2}) must coincide with $\pi$ due to (\ref{coinc}). Now we prove that
 $\{\pi_I\}_{I\in \cI}$
satisfying  compatibly conditions (\ref{comp}) individuates a representation $\pi$ fulfilling (\ref{comp2}).
First of all suppose that there is a nonvanishing pair $(\Phi,\Pi)$ supported in some $I\in \cR$ with
$\pi_I\left(W(\Phi,\Pi)\right)=0$.
Consequently using Weyl relations, for every $J \in \cR$ such that there is $K\in \cR$ with $K \supset I,J$:
$$\pi_J(W(\Phi',\Pi'))=
 \pi_K\left(W(\Phi',\Pi')\right) = c\pi_K\left(W(\Phi'- \Phi,\Pi'-\Pi)\right)\pi_K\left(W(\Phi,\Pi)\right) =0$$
whenever $(\Phi',\Pi') \in \cS_J$, $c \in \bC$ being the appropriate exponential arising by Weyl relations. Taking two such $J$ one easily
concludes that
$\pi_L\left(W(\Phi',\Pi')\right)=0$ for all $L\in \cR$ and $(\Phi',\Pi') \in \cS_L$.
Therefore, by continuity all representations $\pi_I$ are degenerate.
A representation $\pi$ fulfilling (\ref{comp2}) in this case is the degenerate one $\pi(a)=0$ for all $a \in \cW$.
Now consider the case where $\pi_I\left(W(\Phi,\Pi)\right) \neq 0$ unless $(\Phi,\Pi)$ vanishes.
Fix $I,J\in \cR$ with $I\cup J = \bS^1$ and $f,g \in C^\infty(\bS^1,\bR)$
with $f+g =1$ and $\supp f \subset I$, $\supp g \subset J$.
  For $(\Pi,\Phi) \in \cS$ define
  $$\pi\left(W(\Phi,\Pi)\right) \doteq e^{-i\sigma(f(\Phi,\Pi), g(\Phi,\Pi))/2} \pi_I\left(W(f(\Phi,\Pi)\right)
 \pi_J\left(g(\Phi,\Pi))\right)
$$
The right-hand side cannot vanish because all the factors appearing  therein  are invertible by construction.
Making use of (\ref{comp}), it is simply proved that, for every fixed $K \in \cR$
\beq \pi\left(W(\Phi,\Pi)\right) = \pi_K\left(W(\Phi,\Pi)\right) \quad \mbox{for all $(\Phi,\Pi) \in \cS_K$.} \label{QQQ}\eeq
By direct inspection, using Weyl relations one verifies that the nonvanishing operators
$\pi\left(W(\Phi,\Pi)\right)$ fulfills Weyl relations for every $W(\Phi,\Pi) \in \cW$.
Finally consider the sub $C^*$ algebra $\hat{\cW}$ generated in $\cB(\cH)$ from generators
$\pi\left(W(\Phi,\Pi)\right)$.
As is well-known (Bratteli-Robinson 2) there is a faithful representation $\pi$ of $\cW$ onto $\hat{\cW}$
(notice that the unit of $\cW$ is in general represented by an orthogonal projector in $\cB(\cH)$) which uniquely
extends the map $W(\Phi,\Pi) \mapsto \pi\left(W(\Phi,\Pi)\right)$ by linearity and continuity.
By construction (\ref{comp2}) is fulfilled by $\pi$ due to (\ref{QQQ}).
\end{proof}

\end{document}